\documentclass[11pt]{article}
\usepackage{amsfonts}
\usepackage{amsmath}
\usepackage{amssymb}
\usepackage{graphicx}
\usepackage[colorlinks=true,linkcolor=blue,citecolor=red,plainpages=false,pdfpagelabels]{hyperref}
\usepackage{fullpage}
\usepackage{color}
\usepackage[dvipsnames]{xcolor}
\usepackage{mathtools}
\usepackage{bbm}
\usepackage{cite}
\usepackage{float}
\usepackage{array}
\usepackage{verbatim}
\usepackage{appendix}

\usepackage[noblocks]{authblk}

\allowdisplaybreaks

\providecommand{\U}[1]{\protect\rule{.1in}{.1in}}
\newtheorem{theorem}{Theorem}

\newtheorem{lemma}{Lemma}
\newtheorem{corollary}[lemma]{Corollary}

\newtheorem{definition}[lemma]{Definition}

\newtheorem{proposition}[lemma]{Proposition}
\newtheorem{remark}[theorem]{Remark}

\newenvironment{proof}[1][Proof]{\noindent\textbf{#1.} }{\ \rule{0.5em}{0.5em}}

\numberwithin{equation}{section}

\def\Tr{\operatorname{Tr}}

\def\supp{\operatorname{supp}}

\def\>{\rangle}
\def\<{\langle}
\def\({\left(}
\def\){\right)}
\def\[{\left[}
\def\]{\right]}
\def\V{\Vert}
\newcommand{\bra}[1]{\langle #1|}
\newcommand{\ket}[1]{|#1\rangle}

\def\id{\mathsf{id}}

\def\d{\textnormal{d}}
\newcommand{\norm}[1]{\left\lVert#1\right\rVert}

\newcommand{\mc}[1]{\mathcal{#1}}

\DeclareMathOperator*{\SumInt}{%
\mathchoice%
  {\ooalign{$\displaystyle\sum$\cr\hidewidth$\displaystyle\int$\hidewidth\cr}}
  {\ooalign{\raisebox{.14\height}{\scalebox{.7}{$\textstyle\sum$}}\cr\hidewidth$\textstyle\int$\hidewidth\cr}}
  {\ooalign{\raisebox{.2\height}{\scalebox{.6}{$\scriptstyle\sum$}}\cr$\scriptstyle\int$\cr}}
  {\ooalign{\raisebox{.2\height}{\scalebox{.6}{$\scriptstyle\sum$}}\cr$\scriptstyle\int$\cr}}
}

\date{\today}

\begin{document}


\title{Fundamental limits on quantum dynamics based on entropy change}

\author[1]{Siddhartha Das\thanks{\href{mailto:sdas21@lsu.edu}{sdas21@lsu.edu}}}
\author[1]{Sumeet Khatri\thanks{\href{mailto:skhatr5@lsu.edu}{skhatr5@lsu.edu}}}
\author[3]{George Siopsis}
\author[1,2]{Mark M. Wilde}
\affil[1]{Hearne Institute for Theoretical Physics, Department of Physics and Astronomy, Louisiana State University, Baton Rouge, Louisiana, 70803, USA}
\affil[2]{Center for Computation and Technology, Louisiana State University, Baton Rouge, Louisiana, 70803, USA}
\affil[3]{Department of Physics and Astronomy, The University of Tennessee, Knoxville, Tennessee, 37996-1200, USA}

\date{\today}
\maketitle

\begin{abstract}

	
	 
	It is well known in the realm of quantum mechanics and information theory that the entropy is non-decreasing for the class of unital physical processes. However, in general, the entropy does not exhibit monotonic behavior. This has restricted the use of entropy change in characterizing evolution processes. Recently, a lower bound on the entropy change was provided in the work of Buscemi, Das, and Wilde~[Phys.~Rev.~A~93(6),~062314~(2016)]. We explore the limit that this bound places on the physical evolution of a quantum system and discuss how these limits can be used as witnesses to characterize quantum dynamics. In particular, we derive a lower limit on the rate of entropy change for memoryless quantum dynamics, and we argue that it provides a witness of non-unitality. This limit on the rate of entropy change leads to definitions of several witnesses for testing memory effects in quantum dynamics. Furthermore, from the aforementioned lower bound on entropy change, we obtain a measure of non-unitarity for unital evolutions.

\end{abstract}

\section{Introduction}
	Entropy is a fundamental quantity that is of wide interest in physics and information theory \cite{Neu32,Sha48,DGM62,Bek73}. Many natural phenomena are described according to laws based on entropy, like the second law of thermodynamics \cite{Car1824,Men60,Att12}, entropic uncertainty relations in quantum mechanics and information theory \cite{Hir57,Bec75,BM75,MU88,CBTW17}, and area laws in black holes and condensed matter physics \cite{BCH73,BKLS86,Sre93,ECP10}. 
	
	No quantum system can be perfectly isolated from its environment. The interaction of a system with its environment generates correlations between the system and the environment. In realistic situations, instead of isolated systems, we must deal with open quantum systems, that is, systems whose environment is not under the control of the observer. The interaction between the system and the environment can cause a loss of information as a result of decoherence, dissipation, or decay phenomena \cite{Car09,Riv11,Wei12}. The rate of entropy change quantifies the flow of information between the system and its environment.
	
	In this work, we focus on the von Neumann entropy, which is defined for a system in the state $\rho$ as $S(\rho)\coloneqq-\Tr\{\rho\log\rho\}$, and from here onwards we refer to it as the entropy. The entropy is non-decreasing under doubly-stochastic, also called unital, physical evolutions \cite{AU78,AU82}. This has restricted the use of entropy change in the characterization of quantum dynamics only to unital dynamics \cite{Str85,BP02book,CK12,BVMG15,MSW16}. Recently, Ref.~\cite{BDW16} gave a lower bound on the entropy change for any positive trace-preserving map. Lower bounds on the entropy change have also been discussed in Refs.~\cite{Str85,MSW16,AW16,ZHHH01} for certain classes of time evolution. Natural questions that arise are as follows: what are the limits placed by the bound\footnote{Specifically, we consider the bound in Ref.~\cite[Theorem~1]{BDW16} as it holds for arbitrary evolution for both finite- and infinite-dimensional systems.} on the entropy change on the dynamics of a system, and can it be used to characterize evolution processes? 
	
	We delve into these questions, at first, by inspecting another pertinent question: at what rate does the entropy of a quantum system change? Although the answer is known for Markovian one-parameter semigroup dynamics of a finite-dimensional system with full-rank states \cite{Spo78}, the answer in full generality has not yet been given. In Ref.~\cite{KS14}, the result of Ref.~\cite{Spo78} was extended to infinite-dimensional systems with full-rank states undergoing Markovian one-parameter semigroup dynamics (cf., Ref.~\cite{DPR17}). We now prove that the formula derived in Ref.~\cite{Spo78} holds not only for finite-dimensional quantum systems undergoing Markovian one-parameter semigroup dynamics, but also for arbitrary dynamics of both finite- and infinite-dimensional systems with states of arbitrary rank. We then derive a lower bound on the rate of entropy change for any memoryless quantum evolution, also called a quantum Markov process. This lower bound is a witness of non-unitality in quantum Markov processes. Interestingly, this lower bound also helps us to derive witnesses for the presence of memory effects, i.e., non-Markovianity, in quantum dynamics. We compare one of our witnesses to the well-known Breuer-Laine-Piilo (BLP) measure \cite{BLP09} of non-Markovianity for two common examples. As it turns out, in one of the examples, our witness detects non-Markovianity even when the BLP measure does not, while for the other example our measure agrees with the BLP measure. We also provide bounds on the entropy change of a system. These bounds are witnesses of how non-unitary an evolution process is. We use one of these witnesses to propose a measure of non-unitarity for unital evolutions and discuss some of its properties.  

	The organization of the paper is as follows. In Section~\ref{sec:notation}, we introduce standard definitions and facts that are used throughout the paper. In Section~\ref{sec-ent_change_rate}, we discuss the explicit form (Theorem~\ref{thm:rate-oqs}) for the rate of entropy change of a system in any state undergoing arbitrary time evolution. In Section~\ref{sec:qmp}, we provide a brief overview of quantum Markov processes. We state Theorem~\ref{thm:q-Markov-ent-change-rate}, which provides a lower limit on the rate of entropy change for quantum Markov processes. We show that this lower limit provides a witness of non-unitality. We also discuss the implications of the lower limit on the rate of entropy change in the context of bosonic Gaussian dynamics (Section~\ref{sec-Gaussian}). In Section~\ref{sec:qnmp}, based on the necessary conditions for the Markovianity of quantum processes as stated in Theorem~\ref{thm:q-Markov-ent-change-rate}, we define some witnesses of non-Markovianity and also a couple of measures of non-Markovianity based on these witnesses. We apply these witnesses to two common examples of non-Markovian dynamics (Section~\ref{sec-decohere} and Section~\ref{sec-GADC}) and illustrate that they can detect non-Markovianity. In Section~\ref{sec-GADC}, we consider an example of a non-unital quantum non-Markov process whose non-Markovianity goes undetected by the BLP measure while it is detected by our witness. In Section~\ref{sec-entropy_change}, we derive an upper bound on entropy change for unital evolutions. We also show the monotonic behavior of the entropy for a wider class of operations than previously known. In Section~\ref{sec-non_unitarity}, we provide a measure of non-unitarity for any unital evolution based on the bounds on the entropy change obtained in Section~\ref{sec-entropy_change}. We also discuss properties of our measure of non-unitarity. We give concluding remarks in Section~\ref{sec:conclusion}.

\section{Preliminaries}

\label{sec:notation}

	We begin by summarizing some of the standard notations, definitions, and lemmas that are used in Secs. \ref{sec-ent_change_rate}--\ref{sec-non_unitarity}. 
 
	Let $\mc{B}(\mc{H})$ denote the algebra of bounded linear operators acting on a Hilbert space $\mc{H}$, with $\mathbbm{1}_{\mc{H}}$ denoting the identity operator. Let $\dim(\mc{H})$ denote the dimension of $\mc{H}$, and note that this is equal to $+\infty$ in the case that $\mc{H}$ is a separable, infinite-dimensional Hilbert space. If the trace $\Tr\{A\}$ of $A\in\mc{B}(\mc{H})$ is finite, then we call $A$ trace-class. The subset of $\mc{B}(\mc{H})$ containing all trace-class operators is denoted by $\mc{B}_1(\mc{H})$. The subset containing all positive semi-definite operators is denoted by $\mc{B}_{+}(\mc{H})$. We write $P\geq 0$ to indicate that $P\in\mc{B}_{+}(\mc{H})$. Let $\mc{B}_1^{+}(\mc{H})\coloneqq\mc{B}_+(\mc{H})\cap\mc{B}_1(\mc{H})$. Elements of $\mc{B}_1^{+}(\mc{H})$ with unit trace are called density operators, and the set of all density operators is denoted by $\mc{D}(\mc{H})$. The Hilbert space associated with a quantum system $A$ is denoted by $\mc{H}_A$. The state of a quantum system $A$ is represented by a density operator $\rho_A\in\mc{D}(\mc{H}_A)$. We let $\mc{H}_{AB}:=\mc{H}_A\otimes\mc{H}_B$ denote the Hilbert space of a composite system $AB$. The density operator of a composite system $AB$ is denoted by $\rho_{AB}\in\mc{D}(\mc{H}_{AB})$, and the partial trace $\Tr_A$ over the system $A$ gives the local density operator $\rho_B$ of system $B$, i.e., $\rho_B=\Tr_A\{\rho_{AB}\}$. A pure state $\psi_A\coloneqq|\psi\>\<\psi|_A$ is a rank-one element in $\mc{D}(\mc{H}_A)$. 

	Let $\mc{N}_{A\to B}:\mc{B}(\mc{H}_A)\rightarrow\mc{B}(\mc{H}_B)$ denote a linear map (also called superoperator) that maps elements in $\mc{B}(\mc{H}_A)$ to elements in $\mc{B}(\mc{H}_B)$. It is called positive if it maps elements of $\mc{B}_{+}(\mc{H}_A)$ to elements of $\mc{B}_{+}(\mc{H}_B)$ and completely positive if $\id_R\otimes\mc{N}_{A\to B}$ is positive for a Hilbert space $\mc{H}_R$ of any dimension, where $\id$ is the identity superoperator. A positive map $\mc{N}_{A\to B}:\mc{B}_1^{+}(\mc{H}_A)\to\mc{B}_1^{+}(\mc{H}_B)$ is called trace non-increasing if $\Tr\{\mc{N}_{A\to B}(\sigma_A)\}\leq \Tr\{\sigma_A\}$ for all $\sigma_A\in\mc{B}_1^{+}(\mc{H}_A)$, and it is called trace-preserving if $\Tr\{\mc{N}_{A\to B}(\sigma_A)\}=\Tr\{\sigma_A\}$ for all $\sigma_A\in\mc{B}_1^+(\mc{H}_A)$. Where confusion does not arise, we omit identity operators in expressions involving multiple tensor factors, so that, for example, $\mc{N}_{A\to B}(\rho_{RA})$ is understood to mean $\id_{R}\otimes\mc{N}_{A\to B}(\rho_{RA})$.
	
	A linear map $\mc{N}_{A\to B}:\mc{B}(\mc{H}_A)\rightarrow\mc{B}(\mc{H}_B)$ is called sub-unital if $\mc{N}_{A\to B}(\mathbbm{1}_{A})\leq \mathbbm{1}_B$, unital if $\mc{N}_{A\to B}(\mathbbm{1}_{A})= \mathbbm{1}_B$ and super-unital if $\mc{N}_{A\to B}(\mathbbm{1}_{A})\geq\mathbbm{1}_B$, where for $C,D\in\mc{B}(\mc{H})$, $C\geq D$ is defined to mean $C-D\geq 0$. Note that it is possible for a linear map to be neither unital, sub-unital, nor super-unital. A positive trace-preserving map can be sub-unital only if the dimension of the output Hilbert space is greater than or equal to the dimension of the input Hilbert space. A positive trace-preserving map can be super-unital only if the dimension of the output Hilbert space is less than or equal to the dimension of the input Hilbert space. Positive trace-preserving maps between two finite-dimensional Hilbert spaces of the same dimension that are both sub-unital and super-unital are unital.
	
	  The evolution of a quantum state is described by a quantum channel, which by definition is a linear, completely positive, and trace-preserving (CPTP) map. A quantum operation is defined to be a linear, completely positive, and trace non-increasing map. An isometry $U:\mc{H}\to\mc{H}^{\prime}$ is a linear map such that $U^{\dag}U=\mathbbm{1}_{\mathcal{H}}$.
	  

	The adjoint $\mc{N}^\dagger:\mc{B}(\mc{H}_B)\to\mc{B}(\mc{H}_A)$ of a linear map $\mc{N}:\mc{B}_1(\mc{H}_A)\to\mc{B}_1(\mc{H}_B)$ is the unique linear map that satisfies
	\begin{equation}\label{eq-adjoint}
	\forall\ X_A\in\mc{B}_1(\mc{H}_A), Y_B\in\mc{B}(\mc{H}_B):\ \	\<Y_B,\mc{N}(X_A)\>=\<\mc{N}^\dag(Y_B),X_A\>,
	\end{equation}
	where $\<C,D\>=\Tr\{C^\dag D\}$ is the Hilbert-Schmidt inner product. The adjoint of a trace-preserving map is unital, the adjoint of a trace-non-increasing map is sub-unital, and the adjoint of a trace-non-decreasing map is super-unital. 

	
	Let $A$ be a self-adjoint operator acting on a Hilbert space $\mc{H}$. The support $\supp(A)$ of $A$ is the span of the eigenvectors of $A$ corresponding to its non-zero eigenvalues, and the kernel of $A$ is the span of the eigenvectors of $A$ corresponding to its zero eigenvalues. There exists a spectral decomposition of $A$:
	\begin{equation}
		A=\sum_k \lambda_k \ket{k}\bra{k},
	\end{equation}
	where $\{\lambda_k\}_k$ are the eigenvalues corresponding to an orthonormal basis of eigenvectors $\{\ket{k}\}_k$ of $A$. The projection $\Pi(A)$ onto $\supp(A)$ is then
	\begin{equation}
		\Pi(A)=\sum_{k:\lambda_k\neq 0}\ket{k}\bra{k}.
	\end{equation}
	Let $\text{rank}(A)$ denote the rank of $A$. If $A$ is positive definite, i.e., $A>0$, then $\text{rank}(A)=\dim(\mc{H})$, $\Pi(A)=\mathbbm{1}_{\mc{H}}$, and we say that the rank of $A$ is full. If $f$ is a real-valued function with domain $\text{Dom}(f)$, then $f(A)$ is defined as
	\begin{equation}
		f(A)=\sum_{k:\lambda_k\in\text{Dom}(f)}f(\lambda_k)\ket{k}\bra{k}.
	\end{equation}

	The von Neumann entropy of a state $\rho_A$ of a quantum system $A$ is defined as 
	\begin{equation}\label{eq-entropy}
		S(A)_\rho:= S(\rho_A)= -\Tr\{\rho_A\log\rho_A\},
	\end{equation}
	where $\log$ denotes the natural logarithm. In general, the state of an infinite-dimensional quantum system need not have finite entropy \cite{BV13}. For any finite-dimensional system $A$, the entropy is upper-bounded by $\log\dim(\mc{H}_A)$. 
	The quantum relative entropy of any two operators $\rho,\sigma\in\mc{B}^+_1(\mc{H})$ is defined as \cite{Ume62,Fal70,Lin73} 
	\begin{equation}\label{eq-rel_ent_alt}
		D(\rho\V \sigma)=
		\sum_{i,j}\left|\< \phi_i\vert \psi_j\>\right|^2\left[p(i)\log\(\frac{p(i)}{q(j)}\)\right], 
	\end{equation}
	where $\rho=\sum_i p(i)\ket{\phi_i}\bra{\phi_i}$ and $\sigma=\sum_j q(j)\ket{\psi_j}\bra{\psi_j}$ are spectral decompositions of $\rho$ and $\sigma$, respectively, with $\{\ket{\phi_i}\}_i$ and $\{\ket{\psi_j}\}_j$ orthonormal bases for $\mc{H}$. From the above definition, it is clear that $D(\rho \V \sigma) = +\infty$ if $\supp(\rho)\not\subseteq\supp(\sigma)$. Another common way to write the relative entropy is as follows:
	\begin{equation}\label{eq-rel_ent}
		D(\rho\V\sigma)=\left\{\begin{array}{c c} \Tr\{\rho(\log\rho-\log\sigma)\} & \text{if }\supp(\rho)\subseteq\supp(\sigma),\\ +\infty & \text{otherwise}\end{array}\right.
	\end{equation}
	when $\rho(\log\rho-\log\sigma)\in\mc{B}_1(\mc{H})$, where the trace is understood to be with respect to the orthonormal basis $\{\ket{\phi_i}\}_i$. In general, the formula \eqref{eq-rel_ent} has to be evaluated using \eqref{eq-rel_ent_alt}. For any two positive semi-definite operators $\rho$ and $\sigma$, $D(\rho\Vert\sigma)\geq 0$ if $\Tr\{\rho\}\geq\Tr\{\sigma\}$, $D(\rho\Vert\sigma)=0$ if and only if $\rho=\sigma$, and $D(\rho\V\sigma)<0$ if $\rho<\sigma$. The quantum relative entropy is non-increasing under the action of positive trace-preserving maps \cite{MR15}, that is, $D(\rho\V\sigma)\geq D(\mc{N}(\rho)\V\mc{N}{(\sigma)})$ for any two density operators $\rho$ and $\sigma$ and positive trace-preserving map $\mc{N}$. 
	
	The Schatten $p$-norm of an operator $A\in\mc{B}(\mc{H})$ is defined as  
	\begin{equation}
		\Vert A\Vert_p\equiv \left( \Tr\left\{|A|^p\right\}\right)^\frac{1}{p},
	\end{equation}
	where $|A|\equiv \sqrt{A^\dagger A}$ and $p\in[1,\infty)$. If $\{\sigma_i(A)\}_i$ are the singular values of $A$, then 
	\begin{equation}
		\Vert A\Vert_p=\left[\sum_i\sigma_i(A)^p\right]^\frac{1}{p}.
	\end{equation}
	$\norm{A}_{\infty}\coloneqq\lim_{p\to\infty}\norm{A}_p$ is the largest singular value of $A$. Let $\mc{B}_p(\mc{H})$ be the subset of $\mc{B}(\mc{H})$ consisting of operators with finite Schatten $p$-norm.

	\begin{lemma}[H\"{o}lder's inequality\cite{Rog88,Hol89,Bha97}]\label{thm:sch-norm}
		For all $A\in\mc{B}_p(\mc{H})$, $B\in\mc{B}_q(\mc{H})$, and $p,q\in[1,\infty)$ such that $\frac{1}{p}+\frac{1}{q}=1$, it holds that
	\begin{equation}
		\vert\<A,B\>\vert=\left\vert \Tr\left\{A^\dagger B\right\}\right\vert\leq \Vert A\Vert_p\Vert B\Vert_q.
	\end{equation}
	\end{lemma}


	The following important lemma can be found in Ref.~\cite[Corollary~5.2]{W12notes}.
	\begin{lemma}\label{thm:log-con}
		Let $\mc{N}:\mc{B}_+(\mc{H}_A)\to\mc{B}_+(\mc{H}_B)$ be a linear, positive, and sub-unital map. Then, for all $\sigma_A\in\mc{B}_{+}(\mc{H}_A)$ it holds that
		\begin{equation}
			\mc{N}_{A\to B}(\log(\sigma_A))\leq\log(\mc{N}_{A\to B}(\sigma_A)). 
		\end{equation}
	\end{lemma}
	
	We now define entropy change, which is the main focus of our work. 
	
	\begin{definition}[Entropy change]\label{def:ent-change}
		Let $\mc{N}:\mc{B}_1^+(\mc{H})\to\mc{B}_1^+(\mc{H}^\prime)$ be a positive trace-non-increasing map. The entropy change $\Delta S(\rho,\mc{N})$ of a system in the state $\rho\in\mc{D}(\mc{H})$ under the action of $\mc{N}$ is defined as 
		\begin{equation}
			\Delta S(\rho,\mc{N}):=S\(\mc{N}(\rho)\)-S\(\rho\)
		\end{equation}
		whenever $S(\rho)$ and $S(\mc{N}(\rho))$ are finite.
	\end{definition}
	
	It should be noted that $\mc{N}(\rho)$ is a sub-normalized state, i.e., $\Tr\{\mc{N}(\rho)\}\leq 1$, if $\mc{N}$ is a positive trace-non-increasing map. 
	
	It is well known that the entropy change $\Delta S(\rho,\mc{N})$ of $\rho$ is non-negative, i.e., the entropy is non-decreasing, under the action of a positive, sub-unital, and trace-preserving map $\mc{N}$ \cite{AU78,AU82} (see also Refs.~\cite[Section~III]{BDW16} and \cite[Theorem~4.2.2]{Nie02}). Recently, a refined statement of this result was made in Ref.~\cite{BDW16}, which is the following: 

	\begin{lemma}[Lower bound on entropy change]\label{thm:ent-lower-bound}
		Let $\mc{N}:\mc{B}_1^+(\mc{H})\to\mc{B}_1^+(\mc{H}^\prime)$ be a positive, trace-preserving map. Then, for all $\rho\in\mc{D}(\mc{H})$,
		\begin{equation}\label{eq-ent_lower_bound}
			\Delta S(\rho,\mc{N})\geq D\!\(\rho\left\Vert\mc{N}^\dag\circ\mc{N}\(\rho\)\)\right. .
		\end{equation}
	\end{lemma}

	Lemma \ref{thm:ent-lower-bound} gives a tight lower bound on the entropy change. As an example of a map saturating the inequality \eqref{eq-ent_lower_bound}, take the partial trace $\mc{N}_{AB\to B}=\Tr_A$, which is a quantum channel that corresponds to discarding system $A$ from the composite system $AB$. Its adjoint is $\mc{N}^\dagger(\rho_B)=\mathbbm{1}_A\otimes \rho_B$. Then, we have $S\(\mc{N}\(\rho_{AB}\)\)-S\(\rho_{AB}\)=S(\rho_B)-S(\rho_{AB})=D\(\rho_{AB}\Vert \mathbbm{1}_A\otimes\rho_B\)=D\(\rho_{AB}\left\Vert\mc{N}^\dag\circ\mc{N}\(\rho_{AB}\)\)\right.$.


\section{Quantum dynamics and the rate of entropy change}\label{sec-ent_change_rate}

	In general, physical systems are dynamical and undergo evolution processes with time. An evolution process for an isolated and closed system is unitary. However, no quantum system can remain isolated from its environment. There is always an interaction between a system and its environment. The joint evolution of the system and environment is considered to be a unitary operation whereas the local evolution of the system alone can be non-unitary. This non-unitarity causes a flow of information between the system and the environment, which can change the entropy of the system. 

	For any dynamical system with associated Hilbert space $\mc{H}$, the state of the system depends on time. The time evolution of the state $\rho_t$ of the system at an instant $t$ is determined by $\frac{\d\rho_t}{\d t}$ when it is well defined\footnote{By this, we mean that each matrix element of $\rho_t$ is  differentiable with respect to $t$.}. The state $\rho_T$ at some later time $t=T$ is determined by the initial state $\rho_0$, the evolution process, and the time duration of the evolution. Since the time evolution is a physical process, the following condition holds for all $t$:
	\begin{equation}
		\Tr\left\{\dot{\rho}_t\right\}=0,
	\end{equation}
	where $\dot{\rho}_t\coloneqq\frac{\d\rho_t}{\d t}$.
	
	It is known from Refs.~\cite{Spo78,Ber09} that for any finite-dimensional system the following formula for the rate of entropy change holds for any state $\rho_t$ whose kernel remains the same at all times and whose support $\Pi_t$ is differentiable:
	\begin{equation}\label{eq:incorrect-rate}
		\frac{\d}{\d t}S(\rho_t)=-\Tr\left\{\dot{\rho}_t\log\rho_t\right\}.
	\end{equation}
	The above formula has also been applied to infinite-dimensional systems for Gaussian states evolving under a quantum diffusion semigroup \cite{DPR17,KS14} whose kernels do not change in time.
	
	Here, we derive the formula \eqref{eq:incorrect-rate} for states $\rho_t$ having fewer restrictions, which generalizes the statements from Refs.~\cite{Spo78,Ber09}. In particular, we show that the formula \eqref{eq:incorrect-rate} can be applied to quantum dynamical processes in which the kernel of the state changes with time, which can happen because the state has time-dependent support.
	

	\begin{theorem}\label{thm:rate-oqs}
		For any quantum dynamical process with $\dim(\mc{H})<+\infty$, the rate of entropy change is given by
		\begin{equation}\label{eq-pi_ent_change_rate}
			\frac{\d }{\d t}S(\rho_t)=-\Tr\left\{\dot{\rho}_t\log\rho_t\right\},
		\end{equation}
		whenever $\dot{\rho}_t$ is well defined. The above formula also holds when $\dim(\mc{H})=+\infty$ given that $\dot{\rho}_t\log\rho_t$ is trace-class and the sum of the time derivative of the eigenvalues of $\rho_t$ is uniformly convergent\footnote{We define uniform convergence as stated in Ref.~\cite[Definition~7.7]{Rud76}.} on some neighborhood of $t$, however small. 
	\end{theorem}	
	\begin{proof}
		Let $\text{Spec}(\rho_t)$ be the set of all eigenvalues of $\rho_t\in\mc{D}(\mc{H})$, including those in its kernel. Let
		\begin{equation}\label{eq-rho_t_spec}
			\rho_t=\sum_{\lambda(t)\in\text{Spec}(\rho_t)}\lambda(t)P_{\lambda}(t)
		\end{equation}
		be a spectral decomposition of $\rho_t$, where the sum of the projections $P_{\lambda}(t)$ corresponding to $\lambda(t)$ is 
		\begin{equation}
			\sum_{\lambda(t)\in\text{Spec}(\rho_t)}P_{\lambda}(t)=\mathbbm{1}_{\mc{H}}.
		\end{equation}
		The following assumptions suffice to arrive at the statement of the theorem when $\dim(\mc{H})=+\infty$. We assume that $\dot{\rho}_t$ is well defined. We further assume that $\sum_{\lambda(t)\in\text{Spec}(\rho_t)}\dot{\lambda}(t)$ is uniformly convergent on some neighborhood of $t$, and $\dot{\rho}_t\log\rho_t$ is trace-class. We note that when $\dim(\mc{H})<+\infty$, $\sum_{\lambda(t)\in\text{Spec}(\rho_t)}\dot{\lambda}(t)$ and $\dot{\rho}_t\log\rho_t$ are always uniformly convergent and trace-class, respectively.

		Now, we define the function $s:[0,\infty)\times (-1,\infty)\rightarrow (0,\infty)$ by
		\begin{equation}\label{eq-def_s}
			s(t,h)\coloneqq \Tr\{\rho_t^{1+h}\} = \sum_{\lambda(t)\in\text{Spec}(\rho_t)}\lambda(t)^{1+h}.
		\end{equation}
		
		Noting that $\frac{\d}{\d x}a^x=a^x\log a$ for all $a>0$ and $x\in\mathbb{R}$, we have that
		\begin{align}
			\frac{\d}{\d h}\rho_{t}^{h+1}=\rho_{t}^{h+1}\log\rho_{t}.
		\end{align}
		Applying \eqref{eq:app-trace-der-1} in Appendix \ref{app-derivative}, we find that
		\begin{align}
			\frac{\d}{\d t}s(t,h)&=\frac{\d}{\d t}\Tr\{\rho_{t}^{h+1}\}=\left(h+1\right)\Tr\{\rho_{t}^{h}\dot{\rho}_{t}\},\label{eq-st_1}\\
			\frac{\d}{\d h}s(t,h)&=\frac{\d}{\d h}\Tr\{\rho_t^{h+1}\}=\Tr\{\rho_{t}^{h+1}\log\rho_t\}.
		\end{align}
		Then the entropy is
		\begin{equation}
			S(\rho_t)=-\left. \frac{\d}{\d h}s(t,h)\right|_{h=0} = - \Tr\{\rho_t \log  \rho_t\} = - \sum_{\lambda(t)\in\text{Spec}(\rho_t)} \lambda(t) \log\lambda(t),
		\end{equation}
		where by definition $0\log 0=0$.
		
		We note that $\rho_t^h$ is an infinitely differentiable, i.e., a smooth function of $h$, and a differentiable function of $t$ for all $t,h$. Note that the trace is also a continuous function. Since $\frac{\d}{\d h}\frac{\d}{\d t}s(t,h)$ exists and is continuous for all $(t,h)\in[0,\infty)\times (-1,\infty)$, the following exchange of derivatives holds for all $(t,h)\in(0,\infty)\times (-1,\infty)$:
		\begin{equation}
			\frac{\d}{\d h}\left[  \frac{\d}{\d t}s(t,h)\right]=\frac{\d}{\d t}\left[  \frac{\d}{\d h}s(t,h)\right].
		\end{equation}
		This implies that
		\begin{equation}
			\left.\frac{\d}{\d h}\left[  \frac{\d}{\d t}s(t,h)\right]\right|_{h=0}=\frac{\d}{\d t}\left[ \left.\frac{\d}{\d h}s(t,h)\right|_{h=0}  \right]
		\end{equation}
		From \eqref{eq-st_1}, we see that $\frac{\d}{\d t}s(t,h)$ is a smooth function of $h$. Therefore, the Taylor series expansion of this function in the neighborhood of $h=0$ is
		\begin{align}
			\frac{\d}{\d t}s(t,h)&=\left.\frac{\d}{\d t}s(t,h)\right|_{h=0}+\left.\frac{\d}{\d h}\left[\frac{\d}{\d t}s(t,h)\right]\right|_{h=0} h+O(h^2).
		\end{align}
		
		From \eqref{eq-def_s}, we find:
		\begin{align}
			\left.\frac{\d}{\d t}s(t,h)\right|_{h=0}&=\left.\frac{\d}{\d t}\left[\sum_{\lambda(t)\in\text{Spec}(\rho_t)}\lambda(t)^{1+h}\right]\right|_{h=0}=\sum_{\lambda(t)\in\text{Spec}(\rho_t)}\left.\frac{\d}{\d t}\left[\lambda(t)^{1+h}\right]\right|_{h=0}\\
			&=\sum_{\lambda(t)\in\text{Spec}(\rho_t)}\left[ (1+h)\lambda(t)^h\dot{\lambda}(t)\right]_{h=0}\\
			&=\sum_{\lambda(t)\neq 0}\dot{\lambda}(t)\label{eq-st_2}.
		\end{align}
		The second equality follows from Ref.~\cite[Theorem~7.17]{Rud76} due to the uniform convergence of $\sum_{\lambda(t)\in\text{Spec}(\rho_t)}\dot{\lambda}(t)$ on some neighborhood of $t$. To obtain the last equality, we use the following fact: since $\lambda(t)\geq 0$ for all $t$ and $\lambda(t)$ is differentiable, if $\lambda(t^*)=0$ for some time $t=t^*\in (0,\infty)$, then $\dot{\lambda}(t^*)=0$. From \eqref{eq-st_1} and \eqref{eq-st_2}, we obtain
		\begin{equation}\label{eq-st_3}
			\Tr\{\Pi_t\dot{\rho}_t\}=\sum_{\lambda(t)\neq 0}\dot{\lambda}(t)=\frac{\d}{\d t}\sum_{\lambda(t)\neq 0}\lambda(t)=\frac{\d}{\d t}\Tr\{\rho_t\}=0,
		\end{equation}
		where $\Pi_t$ is the projection onto the support of $\rho_t$. The second equality holds because $\dot{\lambda}(t^*)=0$ whenever $\lambda(t^*)=0$ for all $\lambda(t^*)\in\text{Spec}(\rho_{t^*})$ and all $t^*\in (0,\infty)$.
		
		Employing \eqref{eq:app-trace-der-2}, we find:
		\begin{align}
			\frac{\d}{\d h}\left[  \frac{\d}{\d t}s(t,h)\right]
			& =\frac{\d}{\d h}\left[  \left(  h+1\right)  \operatorname{Tr}\{\rho_{t}^{h}%
			\dot{\rho}_{t}\}\right]  \\
			& =\operatorname{Tr}\{\rho_{t}^{h}\dot{\rho}_{t}\}+\left(  h+1\right)
			\operatorname{Tr}\left\{\left[  \rho_{t}^{h}\log\rho_{t}\right]  \dot{\rho}_{t}\right\}.
		\end{align}
		Therefore,
		\begin{align}
			-\frac{\d}{\d t}S(\rho_{t})&=\frac{\d}{\d t}\left[  \left. \frac{\d}{\d h}s(t,h)\right|_{h=0}\right]\\
			&=\left.\frac{\d}{\d h}\left[ \frac{\d}{\d t}s(t,h)\right]\right|_{h=0}\\
			&=\operatorname{Tr}\{\Pi_t\dot{\rho}_{t}\}+
			\operatorname{Tr}\left\{ \dot{\rho}_{t} \Pi_t\log\rho_{t} \right\}\\
			&=\Tr\{\dot{\rho}_t\log\rho_t\},
		\end{align}
		where to obtain the last equality we used \eqref{eq-st_3} and the fact that $\log\rho_t$ is defined on $\supp(\rho_t)$. This concludes the proof. 
	\end{proof}
	
	\bigskip
	
	As an immediate application of Theorem \ref{thm:rate-oqs}, consider a closed system consisting of a system of interest $A$ and a bath (environment) system $E$ in a pure state $\psi_{AE}$, for which the time evolution is given by a unitary $U_{AE}$. Under unitary evolution, the entropy of the composite system $AE$ does not change. Also, for a pure state, the entropy of the composite system is zero, and $S(\rho_{A})=S(\rho_{E})$, where $\rho_{A}$ and $\rho_E$ are the reduced states of the systems $A$ and $E$, respectively. Now, it is often of interest to determine the amount of entanglement in the reduced state $\rho_A$ of the system $A$. Several measures of entanglement have been proposed \cite{PV10}, among which the \textit{entanglement of formation} \cite{BDSW96,Woo01}, the \textit{distillable entanglement} \cite{BDSW96,BBP+96}, and the \textit{relative entropy of entanglement}\cite{VP98,VPRK97} all reduce to the entropy $S(\rho_A)$ of the system $A$ in the case of a closed bipartite system \cite{ON02}. Thus, in this case, the rate of entropy change of the system $A$ is equal to the rate of entanglement change (with respect to the aforementioned entanglement measures) caused by unitary time evolution of the pure state of the composite system, and Theorem \ref{thm:rate-oqs} provides a general expression for this rate of entanglement change.
	
	
	In Appendix \ref{app-rate_ent_change}, we discuss how \eqref{eq-pi_ent_change_rate} generalizes the development in Refs.~\cite{Spo78,Ber09}. We consider examples of dynamical processes in which the support and/or the rank of the state change with time, but the formula \eqref{eq-pi_ent_change_rate} is still applicable according to the above theorem.

\section{Quantum Markov processes}\label{sec:qmp}

	The dynamics of an open quantum system can be categorized into two broad classes, quantum Markov processes and quantum non-Markov processes, based on whether the evolution process exhibits memoryless behavior or has memory effects.
	
	Here, we consider the dynamics of an open quantum system in the time interval $I=[t_1,t_2)\subset\mathbb{R}$ for $t_1<t_2$. We assume that the state $\rho_t\in\mc{D}(\mc{H})$ of the system at time $t\in I$ satisfies the following differential master equation:
	\begin{equation}\label{eq-master_equation}
		\dot{\rho}_t=\mc{L}_t(\rho_t)\quad\forall~t\in I,
	\end{equation}
	where $\mathcal{L}_t$ is called the generator \cite{Kos72}, or Liouvillian, of the dynamics and can in general be time-dependent \cite{Ali07}. A state $\rho_{\text{eq}}$ is called a fixed point, or invariant state of the dynamics, if $\dot{\rho}_{\text{eq}}=0$, or,
	\begin{equation}\label{eq-fixed_point}
		\mc{L}_t(\rho_{\text{eq}})=0\quad\forall~t\in I.
	\end{equation}
	
	In general, the evolution of systems governed by the master equation \eqref{eq-master_equation} is given by the two-parameter family $\{\mc{M}_{t,s}\}_{t,s\in I}$ of maps $\mc{M}_{t,s}:\mc{B}(\mc{H})\rightarrow\mc{B}(\mc{H})$ defined by \cite{Riv11}
	\begin{equation}\label{eq-channel_TO_exp}
		\mc{M}_{t,s}=\mc{T}\exp\left[\int_s^t\mc{L}_{\tau}~\d\tau\right]~~\forall ~s,t\in I,~s\leq t,\quad \mc{M}_{t,t}=\id~~\forall ~t\in I,
	\end{equation}
	where $\mc{T}$ is the time-ordering operator, so that the state $\rho_t$ of the system at time $t$ is obtained from the state of the system at time $s\leq t$ as $\rho_t=\mc{M}_{t,s}(\rho_s)$. The maps $\{\mc{M}_{t,s}\}_{t\geq s}$ satisfy the following composition law:
	\begin{align}
	\forall ~s\leq r\leq t:\ \	\mc{M}_{t,s}&=\mc{M}_{t,r}\circ\mc{M}_{r,s},\label{eq-semi_group_time_dep}\\
	\forall ~t\in I:\ \	\mc{M}_{t,t}&=\id,
	\end{align}
	and in terms of these maps the generator $\mc{L}_t$ is given by
	\begin{equation}\label{eq-generator_time_dep_2}
		\mc{L}_t=\lim_{\varepsilon\to 0^+}\frac{\mc{M}_{t+\varepsilon,t}-\id}{\varepsilon}.
	\end{equation}
	For the maps $\{\mc{M}_{t,s}\}_{t\geq s}$ to represent physical evolution, they must be trace-preserving. This implies that for all $\rho\in\mc{D}(\mc{H})$ the generator $\mc{L}_t$ has to satisfy
	\begin{equation}\label{eq-trace_preserving}
		\Tr\left[\mc{L}_t(\rho)\right]=0\quad\forall ~t\in I.
	\end{equation}

	When the intermediate maps $\mc{M}_{t,r}$ and $\mc{M}_{r,s}$ are positive and trace-preserving for all $s\leq r\leq t$, the condition \eqref{eq-semi_group_time_dep} is called P-divisibility. If the intermediate maps $\mc{M}_{t,r}$ and $\mc{M}_{r,s}$ are CPTP (i.e., quantum channels) for all $s\leq r\leq t$, the condition \eqref{eq-semi_group_time_dep} is called CP-divisibility \cite{WC08,RHP14}. Based on the notion of CP-divisibility, we have the following definition of a quantum Markov process.
	
	\begin{definition}[Quantum Markov process \cite{RHP10}]\label{def:qmp}
		The dynamics of a system in a time interval $I$ are called a quantum Markov process when they are governed by \eqref{eq-master_equation} and they are CP-divisible (i.e., the intermediate maps in \eqref{eq-semi_group_time_dep} are CPTP).
	\end{definition}
	
	An important fact is that the dynamics governed by the master equation \eqref{eq-master_equation} are CP-divisible (hence Markovian) if and only if the generator $\mc{L}_t$ of the dynamics has the Lindblad form 
	\begin{equation}\label{eq-generator_time_dep}
		\mc{L}_t(\rho)=-\iota[H(t),\rho]+\sum_i\gamma_i(t)\left[A_i(t)\rho A_i^\dagger(t)-\frac{1}{2}\left\{A_i^\dagger(t) A_i(t),\rho\right\}\right],
	\end{equation}
	with $H(t)$ a self-adjoint operator and $\gamma_i(t)\geq 0$ for all $i$ and for all $t\in I$. The operators $A_i(t)$ are called Lindblad operators. In the time-independent case, this result was independently obtained by Gorini \textit{et al.} \cite{GKS76} for finite-dimensional systems and by Lindblad \cite{Lin76} for infinite-dimensional systems. For a proof of this result in the time-dependent scenario, see Refs.~\cite{Riv11,CK12}. In finite dimensions, necessary and sufficient conditions for $\mc{L}_t$ to be written in Lindblad form have been given in Ref.~\cite{WEC08}. It should be noted that in general, for some physical processes, $\gamma_i(t)$ can be temporarily negative for some $i$ and the overall evolution still CPTP \cite{LPB10,HCLA14}.
	
	Given the generator $\mc{L}_t$ of the dynamics \eqref{eq-master_equation} and the corresponding positive trace-preserving maps $\{\mc{M}_{s,t}\}_{s,t\in I}$, it holds that the adjoint maps $\{\mc{M}_{s,t}^\dagger\}_{s,t\in I}$ are positive and unital. Furthermore, the adjoint maps $\{\mc{M}_{s,t}^\dagger\}_{s,t\in I}$ are generated by $\mc{L}_t^\dagger$, where $\mc{L}_t^\dagger$ is the adjoint of $\mc{L}_t$. The Lindblad form \eqref{eq-generator_time_dep} of the generator $\mc{L}_t^\dagger$ is
	\begin{equation}\label{eq-generator_adjoint}
		\mc{L}_t^\dagger(X)=\iota[H(t),X]+\sum_i \gamma_i(t)\left(A_i^\dagger(t) X A_i(t)-\frac{1}{2}\left\{X,A_i^\dagger(t) A_i(t)\right\}\right).
	\end{equation}

	Now, let us consider the rate of entropy change $\frac{\d }{\d t}S(\rho_t)$ of a system in state $\rho_t$ at time $t$ evolving under dynamics with Liouvillian $\mc{L}_t$. Theorem~\ref{thm:rate-oqs} implies the following equality:
	\begin{align}
		\frac{\d}{\d t}S(\rho_t)=-\Tr\left\{\mc{L}_t(\rho_t)\log\rho_t\right\}\quad\forall~t\in I.
	\end{align}

	We now derive a limitation on the rate of entropy change of quantum Markov processes using the lower bound in Lemma \ref{thm:ent-lower-bound} on entropy change.

	\begin{theorem}[Lower limit on the rate of entropy change]\label{thm:q-Markov-ent-change-rate}
		The rate of entropy change of any quantum Markov process (Definition \ref{def:qmp}) is lower bounded as
		\begin{equation}\label{eq:qmp-rate-lim}
		\frac{\d}{\d t}S(\rho_t)\geq -\lim_{\varepsilon\to 0^+}\frac{\d}{\d\varepsilon}\Tr\left\{\Pi_t\(\(\mc{M}_{t+\varepsilon,t}\)^\dagger\circ\mc{M}_{t+\varepsilon,t}(\rho_t)\)\right\}= -\Tr\left\{\Pi_t\mc{L}_t^\dag(\rho_t)\right\},
		\end{equation}
		where $\Pi_t$ is the projection onto the support of the state $\rho_t$ of a system.
In general,	\eqref{eq:qmp-rate-lim} also holds for dynamics that obey \eqref{eq-master_equation} and are P-divisible. 		
	\end{theorem}
	\begin{proof}
		First, since the system is governed by \eqref{eq-master_equation}, we have that $\rho_{t+\varepsilon}=\mc{M}_{t+\varepsilon,t}(\rho_t)$ for any $\varepsilon>0$. Also, since $\mc{M}_{t+\varepsilon,t}$ is CPTP (hence positive and trace-preserving), we can use Lemma \ref{thm:ent-lower-bound} to get that
		\begin{equation}
			S(\mc{M}_{t+\varepsilon,t}(\rho_t))-S(\rho_t)\geq D\(\rho_t\left\Vert (\mc{M}_{t+\varepsilon,t})^\dagger\circ\mc{M}_{t+\varepsilon,t}(\rho_t)\)\right.	
		\end{equation}
		Therefore, by definition of the derivative, we obtain
		\begin{align}
			\frac{\d}{\d t}S(\rho_t)&=\lim_{\varepsilon\to 0^+}\frac{S(\rho_{t+\varepsilon})-S(\rho_t)}{\varepsilon}\\
			&\geq \lim_{\varepsilon\to 0^+}\frac{1}{\varepsilon}D\(\rho_t\left\Vert\(\mc{M}_{t+\varepsilon,t}\)^\dagger\circ\mc{M}_{t+\varepsilon,t}(\rho_t)\)\right.\\
			&=\lim_{\varepsilon\to 0^+}\frac{-S(\rho_t)-\Tr\left\{\rho_t\log\left[\(\mc{M}_{t+\varepsilon,t}\)^\dagger\circ\mc{M}_{t+\varepsilon,t}(\rho_t)\right]\right\}}{\varepsilon}\label{eq-lb_pf_1}\\
			&=-\lim_{\varepsilon\to 0^+}\frac{\d}{\d\varepsilon}\Tr\left\{\rho_t\log\left[\(\mc{M}_{t+\varepsilon,t}\)^\dagger\circ\mc{M}_{t+\varepsilon,t}(\rho_t)\right]\right\}\label{eq-lb_pf_2}\\
			&=-\lim_{\varepsilon\to 0^+}\Tr\left\{\rho_t\frac{\d\left(\log\left[\(\mc{M}_{t+\varepsilon,t}\)^\dagger\circ\mc{M}_{t+\varepsilon,t}(\rho_t)\right]\right)}{\d\varepsilon}\right\}\\
			&=-\lim_{\varepsilon\to 0^+}\frac{\d}{\d\varepsilon}\Tr\left\{\Pi_t \(\mc{M}_{t+\varepsilon,t}\)^\dagger\circ\mc{M}_{t+\varepsilon,t}(\rho_t)\right\}\label{eq-lb_pf_3},
		\end{align}
		where we used the definition of the derivative to get \eqref{eq-lb_pf_2} from \eqref{eq-lb_pf_1}. From Appendix \ref{app-derivative}, and noting that $\lim_{\varepsilon\to 0}\(\mc{M}_{t+\varepsilon,t}\)^\dagger\circ\mc{M}_{t+\varepsilon,t}(\rho_t)=\rho_t$, we arrive at \eqref{eq-lb_pf_3}. Then, using the definition of the adjoint and the master equation \eqref{eq-master_equation}, we get
		\begin{align}
			&-\lim_{\varepsilon\to 0^+}\frac{\d}{\d\varepsilon}\Tr\left\{\Pi_t\(\mc{M}_{t+\varepsilon,t}\)^\dagger\circ\mc{M}_{t+\varepsilon,t}(\rho_t)\right\}\nonumber\\
			&\hspace{2.5cm}=-\lim_{\varepsilon\to 0^+}\frac{\d}{\d\varepsilon}\Tr\left\{\mc{M}_{t+\varepsilon,t}(\Pi_t)\mc{M}_{t+\varepsilon,t}(\rho_t)\right\}\\
			&\hspace{2.5cm}=-\lim_{\varepsilon\to 0^+}\Tr\left\{\frac{\d}{\d\varepsilon}\left(\mc{M}_{t+\varepsilon,t}(\Pi_t)\mc{M}_{t+\varepsilon,t}(\rho_t)\right)\right\}\\
			&\hspace{2.5cm}=-\lim_{\varepsilon\to 0^+}\Tr\left\{\left(\frac{\d}{\d\varepsilon}\mc{M}_{t+\varepsilon,t}(\Pi_t)\right)\mc{M}_{t+\varepsilon,t}(\rho_t)+\mc{M}_{t+\varepsilon,t}(\Pi_t)\left(\frac{\d}{\d\varepsilon}\mc{M}_{t+\varepsilon,t}(\rho_t)\right)\right\}.
		\end{align}
		Employing \eqref{eq-generator_time_dep_2} and the fact that $\mc{M}_{t,t}=\id$ for all $t\in I$, we get
		\begin{align}
			\mc{L}_t &=\lim_{\varepsilon\to 0^+}\frac{\mc{M}_{t+\varepsilon,t}-\id}{\varepsilon}= \lim_{\varepsilon\to 0^+}\frac{\d}{\d\varepsilon}\mc{M}_{t+\varepsilon,t}.
		\end{align}
		Therefore,
		\begin{align}
			-\lim_{\varepsilon\to 0^+}\frac{\d}{\d\varepsilon}\Tr\left\{\Pi_t\(\mc{M}_{t+\varepsilon,t}\)^\dagger\circ\mc{M}_{t+\varepsilon,t}(\rho_t)\right\}&=-\Tr\left\{\mc{L}_t(\Pi_t)\rho_t+\Pi_t\mathcal{L}_t(\rho_t)\right\}\\
			&=-\Tr\left\{\Pi_t\mc{L}^\dag_t(\rho_t)\right\}, \label{eq:qmp-final-step}
		\end{align}
		where we used the fact \eqref{eq-st_3} that $\Tr\{\Pi_t\mc{L}_t(\rho_t)\}=\Tr\left\{\Pi_t\dot{\rho}_t\right\}=0$.
	\end{proof}

	\bigskip
	
	Quantum dynamics obeying \eqref{eq-master_equation} are unital in a time interval $I$ if $\mc{L}_t(\mathbbm{1})=0$ for all $t\in I$, which implies that 
	$\Tr\{\Pi_t\mc{L}^\dag_t(\rho_t)\}=0$ for any initial state $\rho_0$ and for all $t\in I$. The deviation of $\Tr\{\Pi_t\mc{L}_t^\dagger(\rho_t)\}$ from zero is therefore a \textit{witness of non-unitality} at time $t$. One can find the maximum deviation of $\left|\Tr\{\Pi_t\mc{L}_t^\dagger(\rho_t)\}\right|$ away from zero by maximizing over all possible initial states and over states at any time $t\in I$ to obtain a measure of non-unitality.   

	\begin{remark}\label{rem:q-Markov-ent-change-rate}
		When $\rho_t>0$, the rate of entropy change of any quantum Markov process is lower bounded as
		\begin{equation}
			\frac{\d}{\d t}S(\rho_t)\geq -\lim_{\varepsilon\to 0}\frac{\d}{\d\varepsilon}\Tr\left\{\(\mc{M}_{t+\varepsilon,t}\)^\dagger\circ\mc{M}_{t+\varepsilon,t}(\rho_t)\right\}= -\Tr\left\{\mc{L}_t^\dag (\rho_t)\right\}.
		\end{equation}
	\end{remark}
	
	Given a quantum Markov process and a state described by a density operator $\rho_t>0$ that is not a fixed (invariant) state of the dynamics, we can make the following statements for $t\in I$ and for all $\varepsilon>0$ such that $[t,t+\varepsilon)\subset I$:
	\begin{itemize}
	\item[(i)] If $\mc{M}_{t+\varepsilon,t}$ is strictly sub-unital, i.e., $\mc{M}_{t+\varepsilon,t}\(\mathbbm{1}\)< \mathbbm{1}$, then its adjoint is trace non-increasing, which means that $\Tr\{\mc{L}^\dag_t(\rho_t)\}<0$. This implies that the rate of entropy change is strictly positive for strictly sub-unital Markovian dynamics. 
	\item[(ii)] If $\mc{M}_{t+\varepsilon,t}$ is unital, i.e., $\mc{M}_{t+\varepsilon,t}\(\mathbbm{1}\)= \mathbbm{1}$, then its adjoint is trace-preserving, which means that $\Tr\{\mc{L}^\dag_t(\rho_t)\}=0$. This implies that the rate of entropy change is non-negative for unital Markovian dynamics. 
	\item[(iii)] If $\mc{M}_{t+\varepsilon,t}$ is strictly super-unital, i.e., $\mc{M}_{t+\varepsilon,\varepsilon}\(\mathbbm{1}\)> \mathbbm{1}$, then its adjoint is trace-increasing, which means that $\Tr\{\mc{L}^\dag_t(\rho_t)\}>0$. This implies that it is possible for the rate of entropy change to be negative for strictly super-unital Markovian dynamics. 
	\end{itemize}
	
	Using the Lindblad form of $\mc{L}_t^\dagger$ in \eqref{eq-generator_adjoint}, we find that 
	\begin{equation}\label{eq-qmp_L_dag}
		\Tr\{\mc{L}_t^\dagger(\rho_t)\}=\sum_i \gamma_i(t)\left<\left[A_i(t),A_i^\dagger(t)\right]\right>_{\rho_t}
	\end{equation}
	where $\<A\>_{\rho}=\Tr\{A\rho\}$. Using this expression, the lower bound on the rate of entropy change for quantum Markov processes when the state $\rho_t>0$ is
	\begin{equation}\label{eq:qmp-rate-lim_full}
		\frac{\d}{\d t}S(\rho_t)\geq \sum_i\gamma_i(t)\left<\left[A_i^\dagger(t),A_i(t)\right]\right>_{\rho_t}.
	\end{equation}
	The inequality \eqref{eq:qmp-rate-lim_full} was first derived in Ref.~\cite{BN88} and recently discussed in Ref.~\cite{OCA17}.
		
\bigskip		
		
	When the generator $\mc{L}_t\equiv\mc{L}$ is time-independent and $I=[0,\infty)$, it holds that the time evolution from time $s\in I$ to time $t\in I$ is determined merely by the time difference $t-s$, that is, $\mc{M}_{t,s}=\mc{M}_{t-s,0}$ for all $t\geq s$. The evolution of the system is then determined by a one-parameter semi-group. We let $\mc{M}_{t}\coloneqq \mc{M}_{t,0}$ for all $t\geq 0$. 
	
	\begin{remark}
		If the dynamics of a system are unital and can be represented by a one-parameter semi-group $\{\mc{M}_t\}_{t\geq 0}$ of quantum channels such that the generator $\mc{L}$ is self-adjoint, then for $\rho_0>0$,
		\begin{equation}
			-\Tr\{\rho_0\log\rho_{2t}\}\leq S(\rho_t)\leq -\Tr\{\rho_{2t}\log\rho_0\}.
		\end{equation}
		This follows from Lemma \ref{thm:log-con}, \eqref{eq-adjoint}, and the fact that $\mc{M}_t^\dagger=\mc{M}_t$. In particular,
		\begin{align}
			S(\rho_t) = S(\mc{M}_t(\rho_0))=-\Tr\{\mc{M}_t(\rho_0)\log\mc{M}_t(\rho_0)\} &\leq -\Tr\{\mc{M}_t(\rho_0)\mc{M}_t(\log\rho_0)\}\\
			&=-\Tr\{\mc{M}_t^\dagger\circ\mc{M}_t(\rho_0)\log\rho_0\}\\
			&=-\Tr\{\rho_{2t}\log\rho_0\}.
		\end{align}
		Similarly,
		\begin{align}
			S(\rho_t) = S(\mc{M}_t(\rho_0))=-\Tr\{\mc{M}_t(\rho_0)\log\mc{M}_t(\rho_0)\}
			& = -\Tr\{\rho_0\mc{M}^\dag_t(\log\mc{M}_t(\rho_0))\} \\
			&\geq -\Tr\{\rho_0\log(\mc{M}^\dag_t\circ\mc{M}_t(\rho_0))\}\\
			&=-\Tr\{\rho_0\log\rho_{2t}\}.
		\end{align}
	\end{remark}
	
	\begin{remark}
		If the dynamics of a system are unital and can be represented by a one-parameter semi-group $\{\mc{M}_t\}_{t\geq 0}$ of quantum channels such that the generator $\mc{L}$ is self-adjoint, then the entropy change is lower bounded as
		\begin{equation}\label{eq-mul}
			S(\rho_t)-S(\rho_0)\geq D(\rho_0\left\Vert \rho_{2t})\right. .
		\end{equation}
		This follows using Lemma \ref{thm:ent-lower-bound}. Under certain assumptions, when the dynamics of a system are described by Davies maps \cite{Dav74}, the same lower bound \eqref{eq-mul} holds for the entropy change \cite{AW16}.
	\end{remark}
	
	From the above remark, we see that the entropy change in a time interval $[0,t]$ is lower bounded by the relative entropy between the initial state $\rho_0$ and the evolved state $\rho_{2t}$ after time $2t$. In the context of information theory, the relative entropy has an operational meaning as the optimal type-II error exponent (in the asymptotic limit) in asymmetric quantum hypothesis testing \cite{HP91,ON00}. The entropy change in the time interval $[0,t]$ is thus an upper bound on the optimal type-II error exponent, where $\rho_0$ is the null hypothesis and $\rho_{2t}$ is the alternate hypothesis.

\subsection{Bosonic Gaussian dynamics}\label{sec-Gaussian}
	
	Let us consider Gaussian dynamics that can be represented by the one-parameter family $\{\mc{G}_t\}_{t\geq 0}$ of phase-insensitive bosonic Gaussian channels $\mc{G}_t$ (cf. Ref.~\cite{HHW09}). It is known that all phase-insensitive gauge-covariant single-mode bosonic Gaussian channels form a one-parameter semi-group \cite{GHLM10}. The Liouvillian for such Gaussian dynamics is time-independent and has the following form:
	\begin{equation}\label{eq-Gaussian_generator}
		\mc{L}=\gamma_+\mc{L}_++\gamma_-\mc{L}_-,
	\end{equation}
	where
	\begin{align}
		\mc{L}_+(\rho)&=\hat{a}^\dagger\rho\hat{a}-\frac{1}{2}\left\{\hat{a}\hat{a}^\dagger,\rho\right\},\\
		\mc{L}_-(\rho)&=\hat{a}\rho\hat{a}^\dagger-\frac{1}{2}\left\{\hat{a}^\dagger\hat{a},\rho\right\},
	\end{align}
	$\hat{a}$ is the field-mode annihilation operator of the system, and the following commutation relation holds for bosonic systems:
	\begin{equation}
		\left[\hat{a},\hat{a}^\dagger\right]=\mathbbm{1}.
	\end{equation}
	The state $\rho_t$ of the system at time $t$ is
	\begin{equation}
		\rho_t=\mc{G}_t(\rho_0)=e^{t\mc{L}}(\rho_0).
	\end{equation}
	The thermal state $\rho_{\text{th}}(N)$ with mean photon number $N$ is defined as
	\begin{equation}
		\rho_{\text{th}}(N)\coloneqq\frac{1}{N+1}\sum_{n=0}^{\infty}\left(\frac{N}{N+1}\right)^n\ket{n}\bra{n},
	\end{equation}
	where $N\geq 0$ and $\{\ket{n}\}_{n\geq 0}$ is the orthonormal, photonic number-state basis. Using \eqref{eq-qmp_L_dag}, we have: 
	\begin{align}
		-\Tr\{\mc{L}^\dagger(\rho_t)\}&=-\gamma_+\left<\left[\hat{a}^\dagger,\hat{a}\right]\right>_{\rho_t}-\gamma_-\left<\left[\hat{a},\hat{a}^\dagger\right]\right>_{\rho_t}\\
		&=\gamma_+-\gamma_-~.
	\end{align}
	Therefore, by Remark \ref{rem:q-Markov-ent-change-rate}, we find that if $\rho_t>0$, then
	\begin{align}
		\frac{\d S(\mc{G}_t(\rho_0))}{\d t}\geq\gamma_+-\gamma_-~.
	\end{align}
	The lower bound $\gamma_+-\gamma_-$ is a witness of non-unitality. It is positive for strictly sub-unital, zero for unital, and negative for strictly super-unital dynamics. For example, when the dynamics are represented by a family $\{\mc{A}_t\}_{t\geq 0}$ of noisy amplifier channels $\mc{A}_t$ with thermal noise $\rho_{\text{th}}(N)$, we have $\gamma_+=N+1$ and $\gamma_-=N$, which implies that the dynamics are strictly sub-unital. When the dynamics are represented by a family $\{\mc{B}_t\}_{t\geq 0}$ of lossy channels $\mc{B}_t$ (i.e., beamsplitters) with thermal noise $\rho_{\text{th}}(N)$, we have $\gamma_+=N$, $\gamma_-=N+1$, which implies that the dynamics are strictly super-unital. When the dynamics are represented by a family $\{\mc{C}_t\}_{t\geq 0}$ of additive Gaussian noise channels $\mc{C}_t$, we have $\gamma_+=\gamma_-$, which implies that the dynamics are unital.

\section{Quantum non-Markov processes}\label{sec:qnmp}

	Dynamics of a quantum system that are not a quantum Markov process as stated in Definition~\ref{def:qmp} are called a quantum non-Markov process. Among these two classes of quantum dynamics, non-Markov processes are not well understood and have attracted increased focus over the past decade. Some examples of applications of quantum Markov processes are in the fields of quantum optics, semiconductors in condensed matter physics, the quantum mechanical description of Brownian motion, whereas some examples where quantum non-Markov processes have been applied are in the study of a damped harmonic oscillator, or a damped driven two-level atom \cite{Car09,Wei12,Riv11}.

	There can be several tests derived from the properties of quantum Markov processes, the satisfaction of which gives witnesses of non-Markovianity. Based on Theorem \ref{thm:q-Markov-ent-change-rate}, we mention here a few tests that will always fail for a quantum Markov process. Passing of these tests guarantees that the dynamics are non-Markovian. 

	An immediate consequence of Theorem \ref{thm:q-Markov-ent-change-rate} is that only a quantum non-Markov process can pass any of the following tests: 
	\begin{itemize}
		\item[(a)] 
			\begin{equation}\label{eq-nmw_sm}
				\frac{\d}{\d t}S(\rho_t) + \lim_{\varepsilon\to 0^+}\frac{\d}{\d\varepsilon}\Tr\left\{\Pi_t\(\(\mc{M}_{t+\varepsilon,t}\)^\dagger\circ\mc{M}_{t+\varepsilon,t}(\rho_t)\)\right\}< 0.
			\end{equation}
		\item[(b)] 
			\begin{equation}\label{eq-nmw_sl}
				\frac{\d}{\d t}S(\rho_t)+\Tr\left\{\Pi_t\mc{L}_t^\dag(\rho_t)\right\}<0.
			\end{equation}
		\item[(c)] 
			\begin{equation}
				\lim_{\varepsilon\to 0^+}\frac{\d}{\d\varepsilon}\Tr\left\{\Pi_t\(\(\mc{M}_{t+\varepsilon,t}\)^\dagger\circ\mc{M}_{t+\varepsilon,t}(\rho_t)\)\right\}\neq \Tr\left\{\Pi_t\mc{L}_t^\dag(\rho_t)\right\}. 
			\end{equation}
	\end{itemize}
	If the dynamics of the system satisfy any of the above tests, then the process is non-Markovian. Based on the description of the dynamics and the state of the system, one can choose which test to apply. In the case of unital dynamics, \eqref{eq-nmw_sm} and \eqref{eq-nmw_sl} reduce to $\frac{\d}{\d t}S(\rho_t)<0$. The observation that the negativity of the rate of entropy change is a witness of non-Markovianity for random unitary processes, which are a particular kind of unital processes, was made in Ref.~\cite{CM14}.

	Based on the above witnesses of non-Markovianity, we can introduce different measures of non-Markovianity for physical processes. Here, we introduce two measures of non-Markovianity that are based on the channel and generator representation of the dynamics of the system:

	\begin{enumerate}
		\item 
			\begin{equation}\label{eq:dual-measure-non-markov}
				\complement_{\textnormal{M}}(\mc{L}):=\max_{\rho_0}\SumInt_{\substack{t:\\\frac{\d S(\rho_t)}{\d t}+\Tr\left\{\Pi_t\mc{L}_t^\dag (\rho_t)\right\}< 0}} \left\vert\frac{\d S(\rho_t)}{\d t}+\Tr\left\{\Pi_t\mc{L}_t^\dag (\rho_t)\right\}
\right\vert.
			\end{equation}
		\item 
			\begin{equation}\label{eq:cor-measure-non-markov}
				\complement_{\textnormal{M}}(\mc{M}):=\max_{\rho_0}\SumInt_{t: f(t)<0} \left\vert f(t)
\right\vert,
			\end{equation}
			where 
			\begin{equation}\label{eq-nm_ft}
				f(t)\coloneqq \frac{\d}{\d t}S(\rho_t)+\lim_{\varepsilon\to 0^+}\frac{\d}{\d\varepsilon}\Tr\left\{\Pi_t\(\(\mc{M}_{t+\varepsilon,t}\)^\dagger\circ\mc{M}_{t+\varepsilon,t}(\rho_t)\)\right\}.
			\end{equation}
	\end{enumerate}
	In the case of unital dynamics, the above measures are equal. It should be noted that the above measures of non-Markovianity are not faithful. This is due to the fact that the statements in Theorem \ref{thm:q-Markov-ent-change-rate} do not provide sufficient conditions for the evolution to be a quantum Markov process. In other words, if the measure $\complement_{\textnormal{M}}$ \eqref{eq:dual-measure-non-markov} is non-zero, then the dynamics are non-Markovian, but if it is equal to zero, then that does not in general imply that the dynamics are Markovian.

\subsection{Examples}

	In this section, we consider two common examples of quantum non-Markov processes: pure decoherence of a qubit system (Section \ref{sec-decohere}) and a generalized amplitude damping channel (Section \ref{sec-GADC}). In order to characterize quantum dynamics, several witnesses of non-Markovianity and measures of non-Markovianity based on these witnesses have been proposed \cite{WEC08,BLP09,RHP10,LWS10,LPP11,LFS12,ZSMWN13,LLW13,LPP13,CM14,HCLA14,HSK15,PGD+16}. Many of these measures are based on the fact that certain quantities are monotone under Markovian dynamics, such as the trace distance between states \cite{BLP09}, entanglement measures \cite{RHP10,LPP11,LFS12}, Fisher information and Bures distance \cite{LWS10,ZSMWN13,LLW13}, and the volume of states \cite{LPP13}. Among these measures, the one proposed in Ref.~\cite{RHP10}  based on the Choi representation of dynamics is both necessary and sufficient. The measure proposed in Ref.~\cite{HCLA14} is also necessary and sufficient and is based on the values of the decay rates $\gamma_i(t)$ appearing in the Lindblad form \eqref{eq-generator_time_dep} of the Liouvillian of the dynamics.
	
	Here, we compare our measures of non-Markovianity with the widely-considered Breuer-Laine-Piilo (BLP) measure of non-Markovianity \cite{BLP09}. This is a measure of non-Markovianity defined using the trace distance and is based on the fact that the trace distance is monotonically non-increasing under quantum channels. Our measure agrees with the BLP measure in the case of pure decoherence of a qubit. In the case of the generalized amplitude damping channel, our witness is able to detect non-Markovianity even when the BLP measure does not. 

\subsubsection{Pure decoherence of a qubit system}\label{sec-decohere}

	Consider a two-level system with ground state $\ket{-}$ and excited state $\ket{+}$. We allow this qubit system to interact with a bosonic environment that is a reservoir of field modes. The time evolution of the qubit system is given by
	\begin{equation}
		\frac{\d\rho_t}{\d t}=-\iota[H(t),\rho_t]+\gamma(t)\left[\sigma_{-}\rho_t\sigma_+-\frac{1}{2}\{\sigma_+\sigma_-,\rho_t\}\right],
	\end{equation}
	where $\sigma_+=\ket{+}\bra{-}$, $\sigma_-=\ket{-}\bra{+}$ and $t\geq 0$. If $H(t)=0$, then the system undergoes pure decoherence and the Liouvillian reduces to
	\begin{equation}
		\mc{L}_t(\rho_t)=\frac{\gamma(t)}{2}\(\sigma_z\rho_t\sigma_z-\rho_t\),
	\end{equation}
	where $\sigma_z=[\sigma_+,\sigma_-]$. The decoherence rate is given by $\gamma(t)$, and it can be determined by the spectral density of the reservoir \cite{BLP09}.
	
	One can verify that $\Tr\{\Pi_t\mc{L}_t^\dagger(\rho_t)\}=0$ for all $t\geq 0$ and any initial state $\rho_0$. This implies that the dynamics are unital for all $t\geq 0$. In this case, for $t>0$, our witness \eqref{eq-nmw_sl} reduces to $\frac{\d}{\d t}S(\rho_t)<0$. For qubit systems undergoing the given unital evolution, it holds that $\rho_t>0$ for all $t>0$, and thus for $t>0$ our measures \eqref{eq:dual-measure-non-markov} and \eqref{eq:cor-measure-non-markov} are equal and reduce to the measure in [Ref.~\cite{HSK15}, Eq.~(15)], which was based on the fact that the rate of entropy change is non-negative for unital quantum channels. As stated therein, these measures of non-Markovianity are positive and agree with those obtained by the BLP measure \cite[Eq.~(11)]{BLP09}.

\subsubsection{Generalized amplitude damping channel}\label{sec-GADC}

	In this example, we consider non-unital dynamics that can be represented as a family of generalized amplitude damping channels $\{\mc{M}_t\}_{t\geq 0}$ on a two-level system \cite{LLW13}. These channels have Kraus operators \cite{Fuj04}
	\begin{equation}\label{eq-GADC}
		\begin{aligned}
		M_t^1&=\sqrt{p_t}\begin{pmatrix} 1&0\\0&\sqrt{\eta_t}\end{pmatrix}\\
		M_t^2&=\sqrt{p_t}\begin{pmatrix} 0&\sqrt{1-\eta_t}\\0&0\end{pmatrix}\\
		M_t^3&=\sqrt{1-p_t}\begin{pmatrix} \sqrt{\eta_t}&0\\0&1\end{pmatrix}\\
		M_t^4&=\sqrt{1-p_t}\begin{pmatrix} 0&0\\\sqrt{1-\eta_t}&0\end{pmatrix},
		\end{aligned}
	\end{equation}
	where $p_t=\cos^2(\omega t)$, $\omega\in\mathbb{R}$, and $\eta_t=e^{-t}$. Then, for all $t\geq 0$, $\mc{M}_t(\rho)=\sum_{i=1}^4 M_t^i\rho (M_t^i)^\dagger$. $\mc{M}_t$ is unital if and only if $p_t=\frac{1}{2}$ or $\eta_t=1$. When $\eta_t=1$, $\mc{M}_t=\id$ for all $\omega$.
	
	It was shown in Ref.~\cite{LLW13} that the BLP measure \cite{BLP09} does not capture the non-Markovianity of the dynamics given by \eqref{eq-GADC}.
	
	Let the initial state $\rho_0$ be maximally mixed, that is, $\rho_0=\frac{1}{2}\mathbbm{1}$. The evolution of this state under $\mc{M}_t$ is then
	\begin{equation}
		\rho_t\coloneqq\mc{M}_t(\rho_0)=\frac{1}{2}\begin{pmatrix} 1+W_t&0\\0&1-W_t\end{pmatrix},
	\end{equation}
	where $W_t=(2p_t-1)(1-\eta_t)=\cos(2\omega t)(1-e^{-t})$. Note that $\rho_t>0$ for all $t\geq 0$. The evolution of these states for an $\varepsilon >0 $ time interval is
	\begin{equation}
		\rho_{t+\varepsilon}=\mc{M}_\varepsilon(\rho_t)=\frac{1}{2}\begin{pmatrix} 1+W_\varepsilon+\eta_\varepsilon W_t&0\\0&1-W_\varepsilon-\eta_\varepsilon W_t\end{pmatrix}
	\end{equation}
	To check whether or not the given dynamics are non-Markovian, let us apply the test in \eqref{eq-nmw_sm}. First, we evaluate
	\begin{equation}
		\mc{M}_\varepsilon^\dagger\circ\mc{M}_\varepsilon(\rho_t)=\frac{1}{2}\begin{pmatrix} a_t&0\\0&b_t\end{pmatrix},
	\end{equation}
	where
	\begin{align}
		a_t&\coloneqq p_t(1+W_\varepsilon+\eta_\varepsilon W_t)+(1-p_t)\eta_t(1+W_\varepsilon+\eta_\varepsilon W_t)+(1-p_t)(1-\eta_t)(1-W_\varepsilon+\eta_\varepsilon W_t)\\
		b_t&\coloneqq p_t\eta_t(1-W_\varepsilon+\eta_\varepsilon W_t)+p_t(1-\eta_t)(1+W_\varepsilon+\eta_\varepsilon W_t)+(1-p_t)(1-W_\varepsilon+\eta_\varepsilon W_t).
	\end{align}
	Then,
	\begin{equation}
		\lim_{\varepsilon\to 0^+}\frac{\d}{\d\varepsilon}\Tr\left\{\mc{M}_\varepsilon^\dagger\circ\mc{M}_\varepsilon(\rho_t)\right\}=W_t.
	\end{equation}
	It should be noted that the deviation of $W_t$ from zero is a witness of non-unitality. For a unital process, for any initial state $\rho_0$ and for all time $t$, we should have $\lim_{\varepsilon\to 0^+}\frac{\d}{\d\varepsilon}\Tr\left\{\Pi_t\mc{M}_\varepsilon^\dagger\circ\mc{M}_\varepsilon(\rho_t)\right\}=0$. For a non-unital process, there will exist some initial state such that for some time $t$, $$\lim_{\varepsilon\to 0^+}\frac{\d}{\d\varepsilon}\Tr\left\{\Pi_t\mc{M}_\varepsilon^\dagger\circ\mc{M}_\varepsilon(\rho_t)\right\}\neq 0.$$
	Next, we evaluate the entropy of the state $\rho_t$ to be
	\begin{equation}
		S(\rho_t)=-\frac{1}{2}\left[(1+W_t)\log\(\frac{1+W_t}{2}\)+(1-W_t)\log\(\frac{1-W_t}{2}\)\right].
	\end{equation}
	This implies that the rate of entropy change is:
	\begin{align}
		\frac{\d S(\rho_t)}{\d t}&=-\frac{1}{2}\frac{\d W_t}{\d t}\log\left[\frac{1+W_t}{2}\right]+\frac{1}{2}\frac{\d W_t}{\d t}\log\left[\frac{1-W_t}{2}\right]\\
		&=\frac{1}{2}\frac{\d W_t}{\d t}\log\left[\frac{1-W_t}{1+W_t}\right],
	\end{align}
	where
	\begin{equation}
		\frac{\d W_t}{\d t}=-2\omega\sin(2\omega t)(1-e^{-t})+\cos(2\omega t)e^{-t}.
	\end{equation}
	Therefore, the test in \eqref{eq-nmw_sm} reduces to
	\begin{equation}\label{eq:f-non-markov}
		f(t)= -\frac{1}{2}\frac{\d W_t}{\d t}\log\left[\frac{1+W_t}{2}\right]+\frac{1}{2}\frac{\d W_t}{\d t}\log\left[\frac{1-W_t}{2}\right]+W_t<0,
	\end{equation}
	where $f$ is defined in \eqref{eq-nm_ft}. For values of $\omega$ such that the dynamics are non-unital, we find that $f$ can be negative in several time intervals; for example, see Fig.~\ref{fig:nu-nm} for the case $\omega=5$.
	
	\begin{figure}[H]
		\centering
		\includegraphics[scale=0.6]{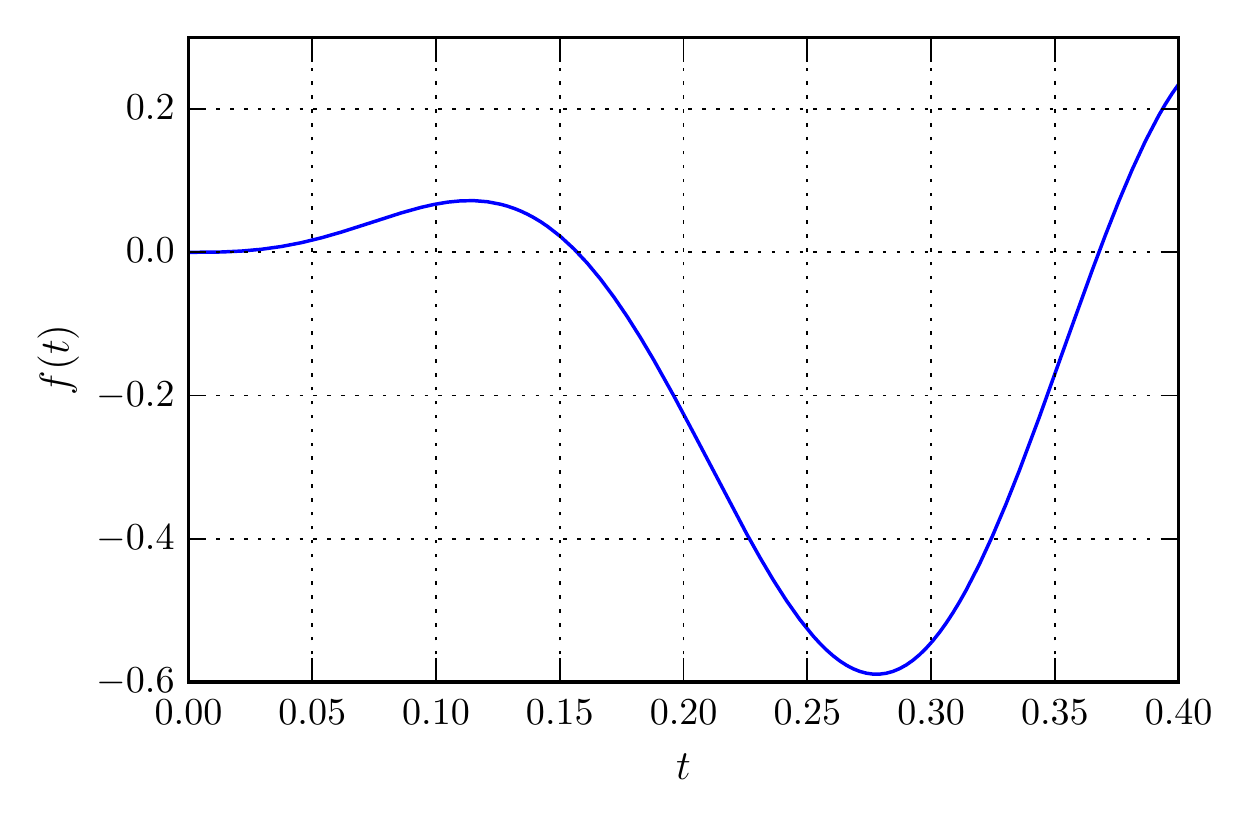}
		\caption{Negative values of $f$, as given in \eqref{eq:f-non-markov}, indicate non-Markovianity. We have taken $\omega=5$.}\label{fig:nu-nm}
	\end{figure}

\section{Bounds on entropy change}\label{sec-entropy_change}

	In this section, we give bounds on how much the entropy of a system can change as a function of the initial state of the system and the evolution it undergoes.
		
	\begin{lemma}\label{thm:ent-lower-bound-2}
		Let $\mc{N}:\mc{B}_1^+(\mc{H})\to\mc{B}_1^+(\mc{H}^\prime)$ be a positive, trace-non-increasing map. Then, for all $\rho\in\mc{D}(\mc{H})$ such that $\mc{N}(\rho)>0$,
		\begin{equation}\label{eq-ent_lower_bound-2}
			\Delta S(\rho,\mc{N})\geq D\(\rho\left\Vert\mc{N}^\dag\circ\mc{N}\(\rho\)\)\right. .
		\end{equation}
	\end{lemma}
	\begin{proof}
		Using the definition \eqref{eq-adjoint} of the adjoint, we obtain
		\begin{align}
			\Delta S(\rho,\mc{N})=S(\mc{N}(\rho))-S(\rho)&=\Tr\{\rho\log\rho\}-\Tr\left\{\mc{N}(\rho)\log\mc{N}(\rho)\right\}\nonumber\\
			&=\Tr\{\rho\log\rho\}-\Tr\left\{\rho\mc{N}^\dag\(\log\mc{N}(\rho)\)\right\}\nonumber\\
			&\geq \Tr\{\rho\log\rho\}-\Tr\left\{\rho \log\[\mc{N}^\dag\circ\mc{N}(\rho)\]\right\}\nonumber\\
			&=D\(\rho\left\Vert\mc{N}^\dag\circ\mc{N}\(\rho\)\)\right. .
		\end{align}
		The inequality follows from Lemma \ref{thm:log-con} applied to $\mathcal{N}^\dagger$, which is positive and sub-unital since $\mathcal{N}$ is positive and trace non-increasing. 
	\end{proof}
	
	\bigskip
	
	Note that for a quantum channel $\mc{N}$, $\Delta S(\rho,\mc{N})=0$ for all $\rho$ if and only if $\rho=\mc{N}^\dag\circ\mc{N}(\rho)$, which is true if and only if $\mc{N}$ is a unitary operation (see Refs.~ \cite[Theorem~2.1]{NS06} and ~\cite[Theorem~3.4.1]{Riv11}). We use this fact to provide a measure of non-unitarity in Section \ref{sec-non_unitarity}.
	
	As an application of the lower bound in Lemma \ref{thm:ent-lower-bound}, let us suppose that a quantum channel $\mc{E}_{A\to B}$ can be simulated as follows
	\begin{equation}\label{eq:env-par}
\forall~\rho_A\in\mc{D}(\mc{H}_A):\		\mc{E}_{A\to B}(\rho_A)=\mc{F}_{AC\to B}(\rho_A\otimes\theta_C),
	\end{equation}
	for a fixed interaction channel $\mc{F}_{AC\to B}$ and a fixed ancillary state $\theta_C$. By applying Lemma \ref{thm:ent-lower-bound} to $\mc{F}$ and the state $\rho_A\otimes\theta_C$, we obtain
	\begin{align}
	\Delta S(\rho_A,\mc{E})&=S\(\mc{F}(\rho_A\otimes\theta_C)\)-S\(\rho_A\) \\
	&\geq S\(\rho_A\otimes\theta_C\)-S\(\rho_A\) +D\(\rho_A\otimes\theta_C\left\Vert \mc{F}^\dagger\circ\mc{F}(\rho_A\otimes\theta_C)\)\right. \\
	&=S\(\theta_C\)+D\(\rho_A\otimes\theta_C\left\Vert \mc{F}^\dagger\circ\mc{F}(\rho_A\otimes\theta_C)\)\right. .
	\end{align}
	Equality holds, i.e., $\Delta S(\rho,\mc{E})=S\(\theta_C\)$, if and only if the interaction channel $\mc{F}$ is a unitary interaction. If $\mc{F}$ is a sub-unital channel, then $\Delta S(\rho,\mc{E})\geq S\(\theta_C\)$ because the relative entropy term is non-negative. This result is of relevance in the context of quantum channels obeying certain symmetries (cf. Ref.~\cite{DW17}).

	\begin{lemma}\label{thm:ent-upper-bound}
		Let $\mc{N}:\mc{B}_+(\mc{H})\to\mc{B}_+(\mc{H}^\prime)$ be a sub-unital channel. Then, for all $\rho\in\mc{D}(\mc{H})$ such that $\rho>0$,
	\begin{equation}\label{eq:ent-upper-bound}
		\Delta S(\rho,\mc{N})\leq \Tr\left\{\[\rho-\mc{N}^\dag\circ\mc{N}(\rho)\]\log\rho\right\}.
	\end{equation}
	This also holds for any positive sub-unital map satisfying the above conditions. 
	\end{lemma}
	\begin{proof}
		By applying Lemma \ref{thm:log-con} to $\mc{N}$, we get
		\begin{align}
			\Delta S(\rho,\mc{N})&=\Tr\{\rho\log\rho\}-\Tr\left\{\mc{N}(\rho)\log\mc{N}(\rho)\right\}\nonumber\\
		&\leq \Tr\{\rho\log\rho\}-\Tr\left\{\mc{N}\(\rho\) \mc{N}\(\log\rho\)\right\}\nonumber\\
			&=\Tr\left\{\[\rho-\mc{N}^\dag\circ\mc{N}(\rho)\]\log\rho\right\}.
		\end{align}
		This concludes the proof.
	\end{proof}

	\bigskip

	By applying H\"{o}lder's inequality (Lemma \ref{thm:sch-norm}) to this upper bound, we obtain the following.
	
	\begin{corollary}\label{thm:sub-u-ent-change}
		Let $\mc{N}:\mc{B}_+(\mc{H})\to\mc{B}_+(\mc{H}^\prime)$ be a sub-unital channel. Then, for all $\rho\in\mc{D}(\mc{H})$ such that $\rho>0$,
		\begin{equation}
			\Delta S(\rho,\mc{N})\leq \left\Vert\rho-\mc{N}^\dag\circ\mc{N}(\rho)\right\Vert_1\left\Vert\log\rho\right\Vert_\infty .
		\end{equation}
	\end{corollary}
	
	\bigskip

	Now, if we let $\mc{N}$ be a sub-unital quantum operation, then as a consequence of Lemma \ref{thm:ent-lower-bound} and Corollary \ref{thm:sub-u-ent-change}, we have, for all states $\rho>0$ such that $\mc{N}(\rho)>0$ and the entropies $S(\rho)$ and $S(\mc{N}(\rho))$ are finite,
	\begin{equation}\label{eq-ent_change_channel}
		D\(\rho\left\Vert\mc{N}^\dag\circ\mc{N}\(\rho\)\) \leq S(\mc{N}(\rho))-S(\rho)\leq \left\Vert\rho-\mc{N}^\dag\circ\mc{N}(\rho)\right\Vert_1\left\Vert\log\rho\right\Vert_\infty \right..
	\end{equation}
	It is interesting to note that \eqref{eq-ent_change_channel} implies
	\begin{equation}
		\norm{\rho-\mc{N}^\dag\circ\mc{N}(\rho)}_1 \geq \frac{1}{\norm{\log\rho}_\infty} D\(\rho\left\Vert\mc{N}^\dag\circ\mc{N}\(\rho\)\)\right.
	\end{equation}
	for a sub-unital quantum operation $\mc{N}$ and a state $\rho>0$ such that $\mc{N}(\rho)>0$. This inequality has the reverse form of Pinsker's inequality \cite{HOT81}, which in this case is
	\begin{equation}\label{eq-pinsker}
		D\(\rho\left\Vert\mc{N}^\dag\circ\mc{N}\(\rho\)\)\right.\geq \frac{1}{2} \norm{\rho-\mc{N}^\dag\circ\mc{N}(\rho)}_1^2.
    \end{equation}
    In general, the relationship between relative entropy and different distance measures, including trace distance, has been studied in Refs.~\cite{BR96,AE05,AE11}.

\section{Measure of non-unitarity}\label{sec-non_unitarity}
	
	In this section, we introduce a measure of non-unitarity for any unital quantum channel that is inspired by the discussion at the end of Section \ref{sec-entropy_change}. A measure of unitarity for channels $\mc{N}:\mc{D}(\mc{H}_A)\to\mc{D}(\mc{H}_A)$, where $\mc{H}_A$ is finite-dimensional, was defined in Ref.~\cite{WGH15}. A related measure for non-isometricity for sub-unital channels was introduced in Ref.~\cite{BDW16}. A measure of non-unitarity for a unital channel is a quantity that gives the distinguishability between a given unital channel with respect to any unitary operation. It quantifies the deviation of a given unital evolution from a unitary evolution. These measures are relevant in the context of cryptographic applications \cite{PLSW04,Aub09} and randomized benchmarking \cite{WGH15}. 

	We know that any unitary evolution is reversible. The adjoint of a unitary operator is also a unitary operator, and a unitary operator and its adjoint are the inverse of each other. These are the distinct properties of any unitary operation. Let us denote a unitary operator by $U_{A\to B}$, where $\dim(\mc{H}_A)=\dim(\mc{H}_B)$. Then a necessary condition for the unitarity of $U_{A\to B}$ is $\(U_{A\to B}\)^\dag U_{A\to B}=\mathbbm{1}_{A}$. The unitary evolution $\mc{U}_{A\to B}$ of a quantum state $\rho_A$ is given by
	\begin{equation}\label{eq-unitary_channel}
		\mc{U}_{A\to B}(\rho_A)= U_{A\to B}(\rho_A) \(U_{A\to B}\)^\dag.
	\end{equation}
	From the reversibility property of a unitary evolution, it holds that $\(\mc{U}_{A\to B}\)^\dag\circ\mc{U}_{A\to B}=\id_A$. It is clear that $\(\mc{U}_{A\to B}\)^\dag$ is also a unitary evolution, and $\(\mc{U}_{A\to B}\)^\dag$ and $\mc{U}_{A\to B}$ are the inverse of each other. 

	Contingent upon the above observation, we note that a measure of non-unitarity for a unital channel $\mc{N}_{A\to B}$ should quantify the deviation of $\(\mc{N}_{A\to B}\)^\dag\circ \mc{N}_{A\to B}$ from $\id_A$ and is desired to be a non-negative quantity. We make use of the trace distance, which gives a distinguishability measure between two positive semi-definite operators and appears in the upper bound\footnote{It should be noted that the lower bound on the entropy change can also be used to arrive at the measure in terms of trace distance by employing Pinsker's inequality \eqref{eq-pinsker}.} on entropy change for a unital channel (Section \ref{sec-entropy_change}), to define a measure of non-unitarity for a unital channel called the \textit{diamond norm of non-unitarity}.    

	\begin{definition}[Diamond norm of non-unitarity]
		The diamond norm of non-unitarity of a unital channel $\mc{N}_{A\to B}$ is a measure that quantifies the deviation of a given unital evolution from a unitary evolution and is defined as
		\begin{equation}\label{eq-od}
			\Vert\mc{N}\Vert_{\oslash}=\norm{\id-\mc{N}^\dagger\circ\mc{N}}_{\diamond},
		\end{equation}
		where the diamond norm $\norm{\cdot}_{\diamond}$ \cite{Kit97} of a Hermiticity-preserving map $\mc{M}$ is defined as
		\begin{equation}\label{eq-diamond_norm}
			\norm{\mc{M}}_{\diamond}=\max_{\rho_{RA}\in\mc{D}(\mathcal{H}_{RA})}\norm{(\id\otimes \mc{M})(\rho_{RA})}_1.
		\end{equation}
		In other words,
		\begin{equation}\label{eq-ot}
			\norm{\mc{N}}_{\oslash}=\max_{\rho_{RA}\in\mc{D}(\mathcal{H}_{RA})}\norm{(\id\otimes(\id-\mc{N}^\dagger\circ\mc{N}))(\rho_{RA})}_1.
		\end{equation}
	\end{definition}

	The diamond norm of non-unitarity of any unital channel $\mc{N}$ has the following properties:
	\begin{enumerate}
		\item $\Vert\mc{N}\Vert_{\oslash}\geq 0$.\label{it-ol}
		\item $\Vert\mc{N}\Vert_{\oslash}=0$ if and only if $\mc{N}^\dag\circ\mc{N}=\id$, i.e., the unital channel $\mc{N}$ is unitary. \label{it-ou}
		\item In \eqref{eq-ot}, it suffices to take $\rho_{RA}$ to be rank one and to let $\dim(\mc{H}_{R})=\dim(\mc{H}_A)$. \label{it-or}
		\item $\norm{\mc{N}}_{\oslash}\leq 2$.\label{it-o2}
	\end{enumerate}
	
	Once we note that $\mc{N}^\dagger\circ\mc{N}:\mc{D}(\mc{H}_A)\to\mc{D}(\mc{H}_A)$ is a quantum channel, properties \ref{it-ol}, \ref{it-or}, and \ref{it-o2} are direct consequences of the properties of the diamond norm \cite{Wat16book}. For property \ref{it-or}, the reference system $R$ has to be comparable with the channel input system $A$, following from the Schmidt decomposition. So $\mc{H}_R$ should be countably infinite if $\mc{H}_A$ is.  Property \ref{it-ou} follows from Refs.~\cite[Theorem~2.1]{NS06} and ~\cite[Theorem~3.4.1]{Riv11}.  

	The diamond norm has an operational interpretation in terms of channel discrimination \cite{W15book,Wat16book} (see also Refs.~\cite{Hel69,Hel76book} for state discrimination). Specifically, the optimal success probability $p_{\text{succ}}(\mc{N}_1,\mc{N}_2)$ of distinguishing between two channels $\mc{N}_1$ and $\mc{N}_2$ is 
	\begin{equation}\label{eq-channel_discrimination}
		p_{\text{succ}}(\mc{N}_1,\mc{N}_2)\coloneqq\frac{1}{2}\left(1+\frac{1}{2}\norm{\mc{N}_1-\mc{N}_2}_{\diamond}\right).
	\end{equation}
	It follows that the optimal success probability of distinguishing between the identity channel and $\mc{N}^\dagger\circ\mc{N}$ is
	\begin{align}
		p_{\text{succ}}(\id,\mc{N}^\dagger\circ\mc{N})&=\frac{1}{2}\left(1+\frac{1}{2}\norm{\id-\mc{N}^\dagger\circ\mc{N}}_{\diamond}\right)\\
		&=\frac{1}{2}\left(1+\frac{1}{2}\Vert\mc{N}\Vert_{\oslash}\right).
	\end{align}

	
	\begin{proposition}\label{prop-unitarity}
		Let $\mc{N}:\mc{D}(\mc{H})\to\mc{D}(\mc{H})$ be a unital channel. If there exists a unitary operator $U\in\mc{B}(\mc{H})$ such that
		\begin{equation}
			\norm{\mc{N}-\mc{U}}_{\diamond}\leq \delta,
		\end{equation}
		where $\mc{U}:\mc{D}(\mc{H})\to\mc{D}(\mc{H})$ is the unitary evolution \eqref{eq-unitary_channel} associated with $U$, then $\norm{\mc{N}}_{\oslash}\leq \sqrt{2\delta}+\delta$.
	\end{proposition}
	\begin{proof}
		We have that
		\begin{align}
			\norm{\id-\mc{N}^\dagger\circ\mc{N}}_{\diamond}&=\norm{\id-\mc{N}^\dagger\circ\mc{U}+\mc{N}^\dagger\circ\mc{U}-\mc{N}^\dagger\circ\mc{N}}_{\diamond}\\
			&\leq\norm{\id-\mc{N}^\dagger\circ\mc{U}}_{\diamond}+\norm{\mc{N}^\dagger\circ(\mc{U}-\mc{N})}_{\diamond}\\
			&\leq\norm{\id-\mc{N}^\dagger\circ\mc{U}}_{\diamond}+\delta.\label{eq-oslash_bound}
		\end{align}
		To obtain these inequalities, we have used the following properties of the diamond norm \cite{Wat16book}:
		\begin{enumerate}
			\item Triangle inequality: $\norm{\mc{N}_1+\mc{N}_2}_{\diamond}\leq\norm{\mc{N}_1}_{\diamond}+\norm{\mc{N}_2}_{\diamond}$.
			\item Sub-multiplicativity: $\norm{\mc{N}_1\circ\mc{N}_2}_\diamond\leq\norm{\mc{N}_1}_{\diamond}\norm{\mc{N}_2}_{\diamond}$.
			\item For all channels $\mc{M}$, $\norm{\mc{M}}_{\diamond}=1$.
		\end{enumerate}
		In particular, to use the third fact, we observe that $\mc{N}^\dagger$ is a channel since $\mc{N}$ is unital. We have also made use of our assumption that $\norm{\mc{U}-\mc{N}}_{\diamond}\leq\delta$.
		
		Now, from the assumption $\norm{\mc{U}-\mc{N}}_{\diamond}\leq\delta$, it follows by unitary invariance of the diamond norm that
		\begin{equation}
			\norm{\id-\mc{U}^\dagger\circ\mc{N}}_{\diamond}\leq\delta.
		\end{equation}
		By the operational interpretation of the diamond distance, this means that the success probability of distinguishing the channel $\mc{U}^\dagger\circ\mc{N}$ from the identity channel, using any scheme whatsoever, cannot exceed $p_{\text{succ}}(\id,\mc{U}^\dagger\circ\mc{N})$ as defined in \eqref{eq-channel_discrimination}. In other words, the success probability cannot exceed $\frac{1}{2}\left(1+\frac{1}{2}\delta\right)$. One such scheme is to send in a bipartite state $\ket{\psi}_{RA}$ on a reference system $R$ and the system $A$ on which the channel acts and perform the measurement defined by the positive operator-valued measure $\{\ket{\psi}\bra{\psi}_{RA},\mathbbm{1}_{RA}-\ket{\psi}\bra{\psi}_{RA}\}$. If the outcome of the measurement is $\ket{\psi}\bra{\psi}_{RA}$, then we guess that the channel is the identity channel, and if the outcome of the measurement is $\mathbbm{1}_{RA}-\ket{\psi}\bra{\psi}_{RA}$ then we guess that the channel is $\mc{U}^\dagger\circ\mc{N}$. The success probability of this scheme is
		\begin{align}
			& \frac{1}{2}\left[  \Tr\{\ket{\psi}\bra{\psi}_{RA}
\id_{RA}(\ket{\psi}\bra{\psi}_{RA})\}\right.\nonumber\\
&\qquad \left.+\Tr\{\left(\mathbbm{1}_{RA}-\ket{\psi}\bra{\psi}_{RA}\right)\left[\id_{R}\otimes(\mathcal{U}^{\dag}\circ\mathcal{N})_{A}\right](\ket{\psi}\bra{\psi}_{RA})\}\right]  \\
			& =\frac{1}{2}\left[  2-\bra{\psi}_{RA}\left[\id_{R}\otimes(\mathcal{U}^{\dag}\circ\mathcal{N})_{A}\right](\ket{\psi}\bra{\psi}_{RA})\ket{\psi}_{RA}\right]  .
		\end{align}
		By employing the above, we find that
		\begin{align}
			\frac{1}{2}\left[2-\bra{\psi}_{RA}\left[\id_R\otimes(\mc{U}^\dagger\circ\mc{N})_A\right](\ket{\psi}\bra{\psi}_{RA})\ket{\psi}_{RA}\right]&\leq\frac{1}{2}\left(1+\frac{1}{2}\delta\right)\\
			\Leftrightarrow \bra{\psi}_{RA}\left[\id_R\otimes(\mc{U}^\dagger\circ\mc{N})_A\right](\ket{\psi}\bra{\psi}_{RA})\ket{\psi}_{RA}&\geq 1-\frac{1}{2}\delta.
		\end{align}
		By employing the definition of the channel adjoint, we find that
		\begin{align}
			&\bra{\psi}_{RA}\left[\id_R\otimes(\mc{U}^\dagger\circ\mc{N})_A\right](\ket{\psi}\bra{\psi}_{RA})\ket{\psi}_{RA} \nonumber \\
			&\qquad =\bra{\psi}_{RA}\left[\id_R\otimes (\mc{N}^\dagger\circ\mc{U})_A\right](\ket{\psi}\bra{\psi}_{RA})\ket{\psi}_{RA}\geq 1-\frac{1}{2}\delta.
		\end{align}
		This holds for all input states, so we can conclude that the following inequality holds:
		\begin{equation}\label{eq-ent_fid}
			\min_{\psi_{RA}}\bra{\psi}_{RA}\left[\id_R\otimes(\mc{N}^\dagger\circ\mc{U})_A\right](\ket{\psi}\bra{\psi}_{RA})\ket{\psi}_{RA}\geq 1-\frac{1}{2}\delta.
		\end{equation}
		Now, by the definition \eqref{eq-diamond_norm} of the diamond norm, and the fact that it suffices to take the maximization in the definition of the diamond norm over only pure states, we have
		\begin{equation}
			\norm{\id-\mc{N}^\dagger\circ\mc{U}}_{\diamond}=\max_{\psi_{RA}}\norm{\left[\id_R\otimes(\id-\mc{N}^\dagger\circ\mc{U})_A\right](\ket{\psi}\bra{\psi}_{RA})}_1.
		\end{equation}
		By the Fuchs-van de Graaf inequality \cite{FV99}, we obtain
		\begin{align}
			&\norm{\left[\id_R\otimes(\id-\mc{N}^\dagger\circ\mc{U})_A\right](\ket{\psi}\bra{\psi}_{RA})}_1\\
			&\qquad =\norm{\ket{\psi}\bra{\psi}_{RA}-\left[\id_R\otimes(\mc{N}^\dagger\circ\mc{U})_A\right](\ket{\psi}\bra{\psi}_{RA})}_1\\
			&\qquad \leq 2\sqrt{1-\bra{\psi}_{RA}\left[\id_R\otimes (\mc{N}^\dagger\circ\mc{U})_A\right](\ket{\psi}\bra{\psi}_{RA})\ket{\psi}_{RA}}.
		\end{align}
		It follows that
		\begin{equation}
			\norm{\id-\mc{N}^\dagger\circ\mc{U}}_{\diamond}\leq 2\sqrt{1-\min_{\psi_{RA}}\bra{\psi}_{RA}\left[\id_R\otimes(\mc{N}^\dagger\circ\mc{U})_A\right](\ket{\psi}\bra{\psi}_{RA})\ket{\psi}_{RA}}.
		\end{equation}
		Using \eqref{eq-ent_fid}, we therefore obtain
		\begin{equation}
			\norm{\id-\mc{N}^\dagger\circ\mc{U}}_{\diamond}\leq 2\sqrt{\frac{1}{2}\delta}=\sqrt{2\delta}.
		\end{equation}
		Finally, from \eqref{eq-oslash_bound} we arrive at
		\begin{equation}
			\norm{\id-\mc{N}^\dagger\circ\mc{N}}_{\diamond}\leq \sqrt{2\delta}+\delta,
		\end{equation}
		as required.
	\end{proof}
	\bigskip
	
	Let us now quantify the non-unitarity of the qudit depolarizing channel $\mc{D}_{d,q}$ defined as \cite{BSST99}
	\begin{equation}
		\mc{D}_{d,q}(\rho)= (1-q)\rho+q\frac{1}{d}\mathbbm{1} \quad\forall \rho\in\mc{D}(\mc{H}_A),
	\end{equation}
	where $\dim(\mc{H}_A)=d$ and $q\in\left[0,\frac{d^2}{d^2-1}\right]$. The input state $\rho$ remains invariant with probability $1-\(1-\frac{1}{d^2}\)q$ under the action of $\mc{D}_{d,q}$.  
	
	\begin{proposition}
		For the depolarizing channel $\mc{D}_{d,q}$, the diamond norm of non-unitarity is 
		\begin{equation}
			\norm{\mc{D}_{d,q}}_{\oslash}=2q(2-q)\left(1-\frac{1}{d^2}\right).
		\end{equation}
	\end{proposition}
	
	\begin{proof}
		The result follows directly from Ref.~\cite[Section~V.A]{MGE12}, but here we provide an alternative proof argument that holds for more general classes of channels.
		
		The depolarizing channel is self-adjoint, that is, $\mc{D}_{d,q}^{\dagger}=\mc{D}_{d,q}$ for all $q$, which means that $\mc{D}_{d,q}^\dagger\circ\mc{D}_{d,q}=\mc{D}_{d,q}^2=\mc{D}_{d,2q-q^2}$. 
	Therefore,
	\begin{equation}\label{eq:norm-d-q}
		\norm{\mc{D}_{d,q}}_{\oslash}=\norm{\id-\mc{D}_{d,q}^2}_{\diamond}=\left|2q-q^2\right|\max_{\psi_{A'A}}\norm{\psi_{A'A}-\psi_{A'}\otimes\frac{\mathbbm{1}}{d}}_1,
	\end{equation}
	where $\psi_{A'A}=\ket{\psi}\bra{\psi}_{A'A}$ is a pure state and $\dim(\mc{H}_{A'})=\dim(\mc{H}_A)=d$. 

	The identity channel and the depolarizing channel  are jointly teleportation-simulable \cite[Definition~6]{DW17} with respect to the resource states, which in this case are the respective Choi states (because these channels are also jointly covariant ~\cite[Definitions~7 \& 12]{DW17}). We know the trace distance is monotonically non-increasing under the action of any channel. Therefore, we can conclude from the form ~\cite[Eq.~(3.2)]{DW17} of the action of jointly teleportation-simulable channels that the diamond norm between any two jointly teleportation-simulable channels is upper bounded by the trace distance between the associated resource states. 
	
	Since $\dim(\mc{H}_A)$ is finite, the maximally entangled state $\ket{\Phi}_{A'A}\coloneqq\frac{1}{\sqrt{d}}\sum_{i=1}^d\ket{i}\ket{i}$, where $\{\ket{i}\}_{i=1}^d$ is any orthonormal basis in $\mc{H}_A$, is an optimal state in \eqref{eq:norm-d-q}. It is known that
	\begin{equation}
		\frac{\mathbbm{1}}{d}\otimes\frac{\mathbbm{1}}{d}=\frac{1}{d^2}\sum_{x=0}^{d^2-1}\sigma_A^x\Phi_{A'A}\sigma_A^x,
	\end{equation}
	where $\{\sigma_A^x\Phi_{A'A}\sigma_A^x\}_{x=0}^{d^2-1}$ forms an orthonormal basis for $\mc{H}_{A'}\otimes\mc{H}_A$ and $\{\sigma^x\}_{x=0}^{d^2-1}$ forms the Heisenberg-Weyl group. We denote the identity element in $\{\sigma^x\}_{x=0}^{d^2-1}$ by $\sigma^0$. Using this, we get
	\begin{align}\label{eq-Dq}
		\norm{\mc{D}_{d,q}}_{\oslash}&=(2q-q^2)\norm{\Phi_{A'A}-\frac{\mathbbm{1}}{d}\otimes\frac{\mathbbm{1}}{d}}_1\\
		&=(2q-q^2)\norm{\(1-\frac{1}{d^2}\)\Phi_{A'A}-\frac{1}{d^2}\sum_{x=1}^{d^2-1}\sigma_A^x\Phi_{A'A}\sigma_A^x}_1\\
		&=(2q-q^2)\left[\(1-\frac{1}{d^2}\)+\frac{d^2-1}{d^2}\right]\\
		&=2q(2-q)\(1-\frac{1}{d^2}\).
	\end{align}
	We conclude that $\norm{\mc{D}_{d,q}}_{\oslash}=2q(2-q)\(1-\frac{1}{d^2}\)$. See Fig.~\ref{fig:non_unitary} for a plot of $\norm{\mc{D}_{2,q}}_{\oslash}$ as a function of $q$.
	\end{proof}
	
	\begin{figure}
		\centering
		\includegraphics[scale=0.6]{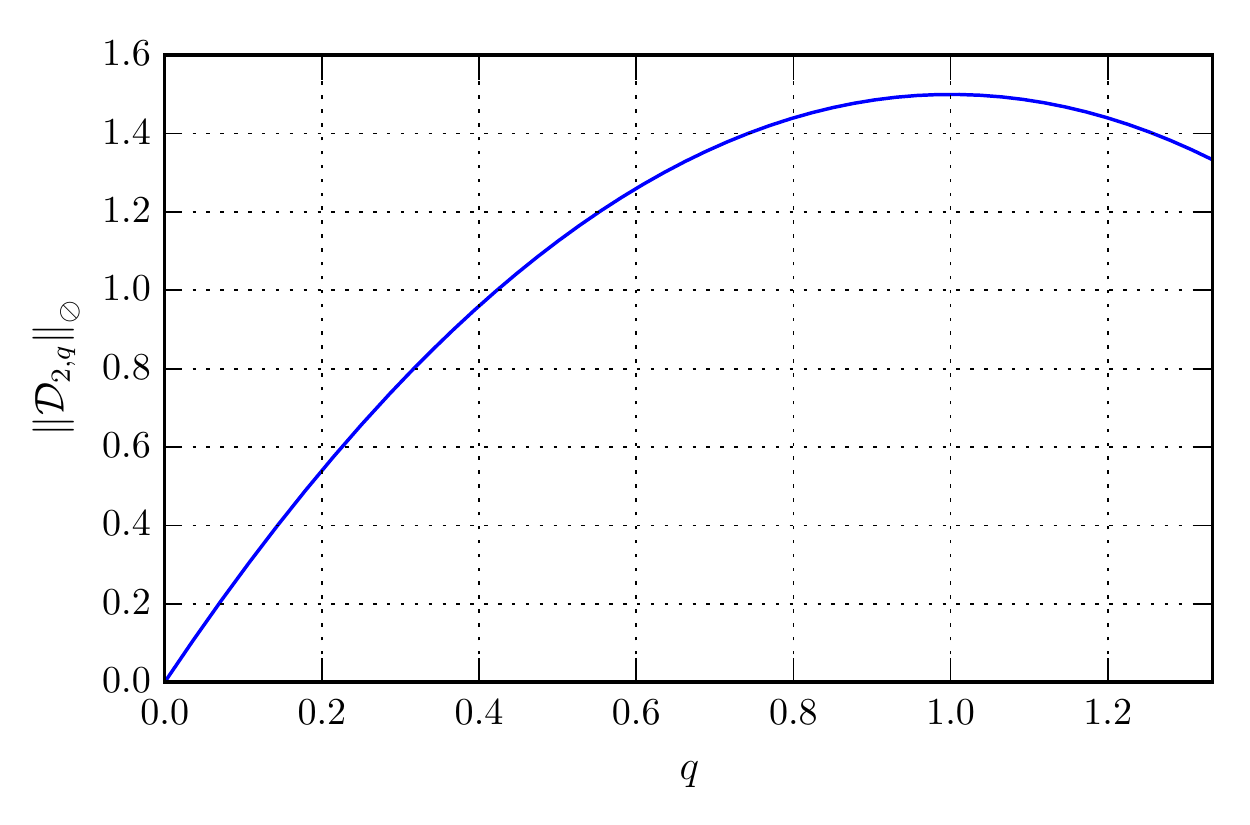}
		\caption{The measure $\norm{\mc{D}_{2,q}}_{\oslash}$ of non-unitarity for the qubit depolarizing channel $\mc{D}_{2,q}$ as a function of the parameter $q\in\left[0,\frac{4}{3}\right]$.}\label{fig:non_unitary}
	\end{figure}

\section{Conclusion}\label{sec:conclusion}

	In this paper, we discussed the rate of entropy change of a system undergoing time evolution for arbitrary states and proved that the formula derived in Ref.~\cite{Spo78} holds for both finite- and infinite-dimensional systems undergoing arbitrary dynamics with states of arbitrary rank. We derived a lower limit on the rate of entropy change for any quantum Markov process. We discussed the implications of this lower limit in the context of bosonic Gaussian dynamics. From this lower limit, we also obtained several witnesses of non-Markovianity, which we used in two common examples of non-Markovian dynamics. Interestingly, in one example, our witness turned out to be useful in detecting non-Markovianity. We generalized the class of operations for which the entropy exhibits monotonic behavior. We also provided a measure of non-unitarity based on bounds on the entropy change, discussed its properties, and evaluated it for the depolarizing channel.

\bigskip
\textbf{Acknowledgments}
We thank Francesco Buscemi, Nilanjana Datta, Jonathan P. Dowling, Omar Fawzi, Eric P. Hanson, Kevin Valson Jacob, Felix Leditzky, Mil\'{a}n Mosonyi, A.~R.~P.~Rau, Cambyse Rouz\'{e}, Punya Plaban Satpathy, and Mihir Sheth for insightful discussions. We are also grateful to Stanislaw Szarek for noticing an error in our previous justification for Proposition \ref{prop-unitarity}. SD acknowledges support from the LSU Graduate School Economic Development Assistantship. SK acknowledges support from the LSU Department of Physics and Astronomy. GS acknowledges support from the U.S. Office of Naval Research under award number N00014-15-1-2646. MMW acknowledges support from the U.S. Office of Naval Research and the National Science Foundation. 


\appendix

\section{Derivatives of operator-valued functions}\label{app-derivative}

	In this section, we recall ~\cite[Theorem~V.3.3]{Bha97}.

	If $f$ is a continuously differentiable function on an open neighbourhood of the spectrum of some self-adjoint operator $A$, then its derivative $Df(A)$ at $A$ is a linear superoperator and its action on an operator $H$ is given by
	\begin{equation}
		Df(A)(H)=\sum_{\lambda,\eta}f^{[1]}(\lambda,\eta)P_A(\lambda)HP_A(\eta),
	\end{equation}
	where $A=\sum_\lambda \lambda P_A(\lambda)$ is the spectral decomposition of $A$ and $f^{[1]}$ is the first divided difference function. 
	
	If $t\mapsto A(t)\in\mc{B}_+(\mc{H})$ is a continuously differentiable function on an open interval in $\mathbb{R}$, with derivative $A'\coloneqq\frac{\d A}{\d t}$, then
	\begin{equation}
		 f'(A(t))\coloneqq \frac{\d}{\d t} f(A(t))= Df(A)(A'(t))=\sum_{\lambda,\eta} f^{[1]}(\lambda,\eta)P_{A(t)}(\lambda)A'(t)P_{A(t)}(\eta).\label{eq:deriv-superop}
	\end{equation}
	In particular, \eqref{eq:deriv-superop} implies the following:
	\begin{align}
		\frac{\d}{\d t}\Tr\{f(A(t))\} & =\Tr\{f'(A(t))A'(t)\},\label{eq:app-trace-der-1}\\
		\Tr\left\{B(t) f'(A(t))\right\} & =\Tr\{B(t)f'(A(t))A'(t)\},\label{eq:app-trace-der-2}
	\end{align}
	where $B(t)$ commutes with $A(t)$.

\section{Rate of entropy change}\label{app-rate_ent_change}

	Here, we continue the discussion from Section \ref{sec-ent_change_rate}. We discuss the subtleties involved in determining the rate of entropy change using the formula \eqref{eq-pi_ent_change_rate} (Theorem \ref{thm:rate-oqs}) by considering some examples of dynamical processes. 
	
	Let us first consider a system in a pure state $\psi_t$ undergoing a unitary time evolution. In this case, the entropy is zero for all time, and thus the rate of entropy change is also zero for all time. Note that even though the rank of the state remains the same for all time, the support changes. This implies that the kernel changes with time. However, $\dot{\psi}_t$ is well defined. This allows us to invoke Theorem \ref{thm:rate-oqs}, so the formula \eqref{eq-pi_ent_change_rate} is applicable. 
	
	Formula \eqref{eq-pi_ent_change_rate} is also applicable to states with higher rank whose kernel changes in time and have non-zero entropy. For example, consider the density operator $\rho_t\in\mc{D}(\mc{H})$ with the following time-dependence:
	\begin{equation}\label{eq:rho-t-i}
		\forall ~t\geq 0:\quad \rho_t=\sum_{i\in\mc{I}} \lambda_i(t) U_i(t)\Pi_i(0)U_i^\dagger(t),
	\end{equation}
	where $\mc{I}=\{i:1\leq i\leq d,~d<\dim(\mc{H})\}$, $\sum_{i\in\mc{I}}\lambda_i(t)=1$, $\lambda_i(t)\geq 0$ and the time-derivative $\dot{\lambda}_i(t)$ of $\lambda_i(t)$ is well defined for all $i\in\mathcal{I}$. The operators $U_i(t)$ are time-dependent unitary operators associated with the eigenvalues $\lambda_i(t)$ such that the time-derivative $\dot{U}_i(t)$ of $U_i(t)$ is well defined and $[U_i(0),\Pi_i(0)]=0$ for all $i\in \mc{I}$. The operators $\Pi_i(0)$ are projection operators associated with the eigenvalues $\lambda_i(0)$ such that the spectral decomposition of $\rho_t$ at $t=0$ is
	\begin{align}\label{eq:rho-0-i}
		\rho_0=\sum_{i\in\mc{I}} \lambda_i(0)\Pi_i(0),
	\end{align}
	where $1<\text{rank}(\rho_0)<\dim(\mc{H})$. The evolution of the system is such that $\text{rank}(\rho_t)=\text{rank}(\rho_0)$ for all $t\geq 0$. It is clear from \eqref{eq:rho-t-i} and \eqref{eq:rho-0-i} that the projection $\Pi_{t}$ onto the support of $\rho_t$ depends on time:
	\begin{equation}
		\Pi_{t}=\sum_{i\in \mc{I}}U_i(t)\Pi_i(0)U_i^\dagger(t),
	\end{equation}
	and the time-derivative $\dot{\Pi}_{t}$ of $\Pi_t$ is well defined. The entropy of the system is zero if and only if the state is pure.  
	
	Let us consider a qubit system $A$ undergoing a damping process such that its state $\rho_t$ at any time $t\geq 0$ is as follows:  
	\begin{equation}
		\rho_t=(1-e^{-t})\ket{0}\bra{0}+e^{-t}\ket{1}\bra{1},
	\end{equation}
	where $\{\ket{0},\ket{1}\}$ is a fixed orthonormal basis of $\mc{H}_A$. The entropy $S(\rho_t)$ of the system at time $t$ is
	\begin{equation}\label{eq-rec1}
		S(\rho_t)=-(1-e^{-t})\log(1-e^{-t})-e^{-t}\log(e^{-t}),
	\end{equation} 
	which is continuously differentiable for all $t>0$ and not differentiable at $t=0$. At $t=0$, $\Pi_0=\ket{1}\bra{1}$ and $\text{rank}(\rho_0)=1$. At $t=0^+$, there is a jump in the rank from 1 to 2, and the rank and the support remains the same for all $t\in(0,\infty)$. In this case, the formula \eqref{eq-pi_ent_change_rate} agrees with the derivative of \eqref{eq-rec1}.
	
	Now, suppose that the system $A$ undergoes an oscillatory process such that for any time $t\geq 0$ the state $\rho_t$ of the system is given by
	\begin{equation}
		\rho_t=\cos^2(\pi t)\ket{0}\bra{0}+\sin^2(\pi t)\ket{1}\bra{1}.
	\end{equation}
	In this case, for all $t\geq 0$, the entropy $S(\rho_t)$ is
	\begin{equation}
		S(\rho_t)=-\cos^2(\pi t)\log\cos^2(\pi t)-\sin^2(\pi t)\log\sin^2(\pi t),
	\end{equation}
	and its derivative is
	\begin{equation}\label{eq-example2_derivative}
		\frac{\d}{\d t}S(\rho_t)=\pi\sin(2\pi t)\left[\log\cos^2(\pi t)-\log\sin^2(\pi t)\right],
	\end{equation}
	which exists for all $t\geq 0$. At $t=\frac{n}{2}$ for all $n\in\mathbb{Z}^+\cup\{0\}$, there is a jump in the rank from 1 to 2 and the support changes discontinuously at these instants. One can check that \eqref{eq-pi_ent_change_rate} and \eqref{eq-example2_derivative} are in agreement for all $t\geq 0$.


\bibliographystyle{alpha}
\bibliography{entropy}


\end{document}